%% file: main.tex
\documentclass[aps,pra,reprint,longbibliography,superscriptaddress,floatfix]{revtex4-2}

\input{header}

\makeatletter
\def\frontmatter@abstractheading{}
\def\frontmatter@abstract@produce{%
  \par
  \addvspace{\frontmatter@preabstractspace}%
  \begingroup
    \prep@absbox
    \unvbox\absbox
    \post@absbox
  \endgroup
  \@ifx{\@empty\mini@notes}{}{\mini@notes\par}%
  \addvspace\frontmatter@postabstractspace
}
\makeatother

\begin{document}
\sloppy

\title{\Title}

\author{Samuel Godwood}
\email{s.godwood@liverpool.ac.uk}
\affiliation{Department of Physics, University of Liverpool, Liverpool, L69 7ZL, United Kingdom}

\author{Do\~{g}a Murat K\"{u}rk\c{c}\"{u}o\~{g}lu}
\affiliation{Fermi National Accelerator Laboratory, Batavia, Illinois 60510, USA}

\author{Gabriel N. Perdue}
\affiliation{Fermi National Accelerator Laboratory, Batavia, Illinois 60510, USA}

\author{Marina Maneyro}
\affiliation{Department of Physics, University of Liverpool, Liverpool, L69 7ZL, United Kingdom}

\author{Alessandro Roggero}
\affiliation{Dipartimento di Fisica, University of Trento, Via Sommarive 14, I-38123 Povo, Trento, Italy}
\affiliation{INFN-TIFPA Trento Institute of Fundamental Physics and Applications, Via Sommarive 14, I-38123 Povo, Trento, Italy}

\begin{abstract}
Finite local Hilbert-space truncations arise naturally in quantum simulations of lattice field theories and motivate qudit encodings, but their fault-tolerant advantage over qubit encodings remains unclear. We compare the non-Clifford cost of implementing quadratic diagonal evolutions, exemplified by $U=e^{-it\phi_x^2}$ in a uniform field-amplitude discretization of a real scalar field, using either one logical $d$-level qudit or $n_b=\lceil \log_2 d\rceil$ logical qubits. We analyze two standard settings: product-formula simulation and LCU/block encoding, taking the resource metric to be the number of non-Clifford gates after synthesis into a discrete logical gate set. Because tight synthesis bounds for general single-qudit rotations are not known, we express the qudit constructions in terms of embedded two-level $SU(2)$ rotations and derive explicit finite-$d$ break-even conditions for their synthesis cost; these serve as compiler targets for when qudit encodings can outperform the qubit baseline. Within the constructive models studied here, product-formula implementations would require an exponentially stronger per-primitive synthesis advantage for qudits to win asymptotically, while in the LCU setting the qubit encoding is asymptotically cheaper in $d$. Nevertheless, the finite-$d$ threshold analysis identifies low dimensional regions in which qudits can yield meaningful constant-factor savings, particularly for LCU-based implementations. As a secondary analysis of the LCU construction, we use an idealized negligible-overhead qubit-qudit code-switching model to give an absolute $T$-count comparison, and reinterpret the savings as an allowable per-switch overhead budget.
\end{abstract}

\maketitle

\input{body}
\clearpage

\appendix
\input{appendix}
\clearpage

\makeatletter
\immediate\write\@auxout{\string\citation{apsrev42Control}}
\makeatother

\bibliographystyle{apsrev4-2}
\bibliography{refs}

\end{document}

%% file: header.tex
\usepackage[utf8]{inputenc}

\usepackage{graphicx}
\usepackage{hyperref}
\usepackage{amsmath}
\usepackage{amssymb}
\usepackage{amsthm}
\usepackage[braket]{qcircuit}
\usepackage{booktabs}

\hypersetup{
    colorlinks,
    linkcolor={blue},
    citecolor={blue},
    urlcolor={blue}
}

\newcommand{\Title}{Fault-Tolerant Resource Comparison of Qudit and Qubit Encodings for Diagonal Quadratic Operators}

\theoremstyle{plain}
\newtheorem{theorem}{Theorem}

\newtheorem{lemma}[theorem]{Lemma}
\newtheorem{corollary}[theorem]{Corollary}

\theoremstyle{definition}
\newtheorem{definition}[theorem]{Definition}

\theoremstyle{remark}


%% file: body.tex
\section{Introduction}\label{sec:intro}

Quantum computers promise first-principles access to strongly correlated quantum dynamics
that are difficult to treat classically,
with applications spanning high-energy physics and lattice gauge theory,
condensed matter and materials, and quantum chemistry
\cite{feynman82,lloyd96,preskill_qft_2018,
georgescu_quantum_simulation_2014,
cao_quantum_chemistry_2019,
ma_quantum_materials_2020}.
In particular, real-time Hamiltonian dynamics and finite-density settings
can challenge standard classical methods,
including regimes where Monte Carlo approaches suffer from severe sign problems
\cite{troyer_wiese_signproblem_2005,gattringer_signproblem_review_2016,nagata_finite_density_lqcd_review_2021}.
This has motivated a broad program of quantum algorithms for Hamiltonian simulation,
spanning product-formula, qubitization/quantum signal processing,
and linear-combination-of-unitaries (LCU)
\cite{trotter59,suzuki90,childs_product_formula_2019,childs_wiebe_lcu_2012,berry_truncated_taylor_2015,low_chuang_qubitization_2019,gilyen_qsvt_blockencoding_2019}.
Concrete proposals exist for simulating scalar field theories, gauge theories,
and related lattice models using digital quantum computers
\cite{jordan_lee_preskill_scattering_2012,preskill_qft_2018,klco_savage_scalar_digitization_2018,zohar_farace_reznik_cirac_digital_lgt_2017,bender_zohar_digital_lgt_3d_2018,klco_stryker_su2_digital_2020,ji_lamm_zhu_group_decimation_2020,illa_robin_savage_qu8its_2024}.

A recurring technical feature of these targets is the presence of local degrees of freedom
that are naturally infinite-dimensional (bosonic modes) or continuous/compact (gauge links),
requiring a truncation or digitization to a finite local Hilbert space to obtain a finite model
\cite{kogut_susskind_hamiltonian_lgt_1975,chandrasekharan_wiese_quantum_link_1997,zohar_formulation_lgt_qsim_2015,banuls_lgt_qtech_2020,ji_lamm_zhu_group_decimation_2020}.
For scalar fields, one common approach is a field-amplitude (position-basis) digitization
on a uniform grid \cite{jordan_lee_preskill_scattering_2012,klco_savage_scalar_digitization_2018,macridin_spentzouris_amundson_harnik_sampling_2018};
for gauge links and rotors, one often imposes representation cutoffs
or angular-momentum truncations
\cite{kogut_susskind_hamiltonian_lgt_1975,chandrasekharan_wiese_quantum_link_1997,zohar_formulation_lgt_qsim_2015,zohar_cirac_reznik_ultracold_gauge_2013,ji_lamm_zhu_group_decimation_2020}.
Importantly, physically relevant truncations can be modest
in some parameter regimes and model families---for example,
finite-group and quantum-link formulations, few-qubit Schwinger-model
realizations, and qudit gauge-field encodings with local dimensions such
as \(d=3\) and \(d=8\) 
\cite{chandrasekharan_wiese_quantum_link_1997,
zohar_farace_reznik_cirac_digital_lgt_2017,
martinez_schwinger_2016,
surace_rydberg_gauge_string_2020,
gonzalez_cuadra_nonabelian_qudits_2022,
illa_robin_savage_qu8its_2024}---making the structure of the finite-\(d\)
encoding consequential for end-to-end algorithmic cost.

These observations have motivated sustained interest in qudit-based quantum computation,
where the native information carrier is a \(d\)-level system rather than a qubit
\cite{wang_qudits_review_2020,gottesman_qudit_ft_1998,hostens_dehaene_de_moor_qudit_stabilizer_2005,farinholt_clifford_qudits_2014,campbell_anwar_browne_msd_prime_2012}.
From an algorithmic perspective, qudits can offer a more direct representation
of truncated local Hilbert spaces
\cite{wang_qudits_review_2020,illa_robin_savage_qu8its_2024,gonzalez_cuadra_nonabelian_qudits_2022},
and in some settings can reduce circuit depth or simplify arithmetic
by exploiting generalized Clifford structure
\cite{gottesman_qudit_ft_1998,hostens_dehaene_de_moor_qudit_stabilizer_2005,farinholt_clifford_qudits_2014,wang_qudits_review_2020}.
More broadly, qudit-based information processing has grown into an active area spanning gate-decomposition techniques, algorithm design, and quantum error correction, as surveyed in recent reviews~\cite{wang_qudits_review_2020,kiktenko_qudits_review_2025}.
From an implementation perspective, several leading hardware platforms
provide access to well-controlled multi-level structure,
enabling qudit gates or encoded qudits within a larger physical Hilbert space
\cite{ringbauer_universal_qudit_processor_2022,neeley_superconducting_phase_qudit_2009,nguyen_multidimensional_solid_state_qudit_2024,erhard_high_dim_entanglement_review_2020}.

Prevailing expectations for quantum advantage typically rely on fault-tolerant computation,
where error correction and compilation overhead dominate
\cite{bravyi_kitaev_msd_2005,fowler_surface_code_review_2012,litinski_game_surface_codes_2019,katabarwa_early_ftqc_2024}.
In this regime, resource estimates are often governed by the cost
of implementing non-Clifford operations in a discrete logical gate set
(e.g., via magic-state distillation)
\cite{bravyi_kitaev_msd_2005,ross_selinger_optimal_rz_2014,kliuchnikov_maslov_mosca_exact_2013,litinski_game_surface_codes_2019}.
While there is an extensive literature on qudit algorithms
and near-term qudit simulation prospects
\cite{wang_qudits_review_2020,banuls_lgt_qtech_2020,zohar_formulation_lgt_qsim_2015,gonzalez_cuadra_nonabelian_qudits_2022,illa_robin_savage_qu8its_2024,kurkcuoglu_lamm_maestri_qudit_decomp_2024, joshi2025efficientquditcircuitquench, joshi2025probinghadronscatteringlattice},
systematic fault-tolerant resource comparisons that track non-Clifford resources
across qudit and qubit encodings remain less standardized,
especially at general \(d\).

In this paper, we give a fault-tolerant resource analysis
for a representative and widely occurring problem class: onsite quadratic diagonal evolutions,
exemplified by $e^{-it\phi_x^2}$ for a digitized scalar field $\phi_x$
in the field-amplitude basis in position space
\cite{jordan_lee_preskill_scattering_2012,klco_savage_scalar_digitization_2018,preskill_qft_2018}.
While our primary presentation uses $\phi_x$ for concreteness, 
the fault-tolerant qubit--qudit comparison established in this work concerns onsite quadratic diagonal unitaries of the form $e^{-it\,g(N)}$. Here, $g$ is quadratic in a truncated integer-valued local operator $N$ (e.g., a number operator), as arises in lattice Hamiltonians \cite{kogut_susskind_hamiltonian_lgt_1975,zohar_formulation_lgt_qsim_2015,banuls_lgt_qtech_2020}.

We study two standard simulation paradigms.
In Regime 1, we compare a single product-formula (Trotter) step
implementing the onsite diagonal evolution using either one $d$-level qudit acting natively on the $d$-level truncation, or $n_b=\lceil \log_2 d\rceil$ qubits in the standard binary embedding of the same physical truncation
\cite{trotter59,suzuki90,childs_product_formula_2019}.
In Regime 2, we compare block-encodings in the LCU framework,
where the cost is set by the PREPARE/SELECT oracle implementations
\cite{childs_wiebe_lcu_2012,berry_truncated_taylor_2015,gilyen_qsvt_blockencoding_2019,low_chuang_qubitization_2019}.
Across both regimes, our metric of fault-tolerant cost is the non-Clifford count
after synthesis into a discrete logical gate set;
synthesis assumptions enter only through the cost of approximating
continuous one-parameter primitives (qubit $R_z$ rotations
and embedded two-level qudit $SU(2)$ rotations)
to a target precision $\varepsilon$
\cite{ross_selinger_optimal_rz_2014,kliuchnikov_maslov_mosca_exact_2013}.
Accordingly, the results should be interpreted as a compilation/resource-model comparison under a chosen fault-tolerant proxy, and not as a complete architecture-level comparison of total logical qudit versus logical qubit overhead.
This would additionally require architecture-specific assumptions about code conversion, state injection/distillation, and space–time costs.

Because tight, gate-set--dependent synthesis constants
for the required single-qudit embedded $SU(2)$ primitives
are not known uniformly in $d$,
we cannot claim closed-form costs for $d>3$.
Instead, we parameterize qudit synthesis by a logarithmic prefactor $a$
and report explicit break-even conditions
$a \lesssim a_{\max}^{\mathrm{PF}}(d,\varepsilon)$ in Regime~1
and $a \lesssim a_{\max}^{\mathrm{LCU}}(d,t,\varepsilon_{\mathrm{sim}})$ in Regime~2.
$a_{\max}^{\mathrm{PF}}(d,\varepsilon)$ and $a_{\max}^{\mathrm{LCU}}(d,t,\varepsilon_{\mathrm{sim}})$ are break-even compiler targets: they are the largest allowed values of the logarithmic synthesis prefactor $a$ of the relevant embedded two-level qudit primitives for which the qudit construction still uses fewer non-Clifford gates than the stated qubit baseline under our cost model, so a smaller value denotes a more demanding compilation requirement. They are formally defined and analyzed in Secs.~\ref{sec:regime1_crossover} and~\ref{sec:regime2_nocs_crossover}, respectively.

With this caveat, our contributions are as follows:
(1) we derive explicit finite-$d$ break-even synthesis thresholds for the
required embedded two-level qudit $SU(2)$ primitives in both regimes;
(2) we place these thresholds in context with asymptotic results under the
explicit models studied here (see Table~\ref{tab:asymptotic_summary}): in Regime~1, we show that asymptotic
qudit advantage would require an
exponentially strong per-primitive synthesis improvement over the qubit
baseline, while in Regime~2 we derive asymptotic non-Clifford scalings, for which the qubit
implementation is asymptotically cheaper; 
(3) as a secondary diagnostic, under an idealized free qubit-qudit code-switching assumption in Regime 2, we translate the resulting finite-$d$ savings into a conditional per-switch overhead budget for a hybrid route.

The asymptotic comparisons of contribution~(2) are unconditional 
consequences of the explicit constructions and cost models studied here: within these models the qubit encoding is
asymptotically no more expensive in non-Clifford count in both regimes, and in Regime~2 this holds for any fixed 
positive synthesis prefactor. 
By contrast, the finite-$d$ break-even results of contribution~(1), 
expressed through the thresholds $a_{\max}^{\mathrm{PF}}(d,\varepsilon)$ and $a_{\max}^{\mathrm{LCU}}(d,t,\varepsilon_{\mathrm{sim}})$, 
are conditional compiler-target statements: whether a finite-$d$ qudit advantage is realized depends on the achievable synthesis cost of
the embedded two-level qudit primitives, which is not yet known uniformly in $d$. The code-switching comparison of contribution~(3) 
is further conditional on an idealized negligible-overhead qubit--qudit code-switching model, and is therefore stated only as an 
allowable per-switch overhead budget rather than a demonstrated advantage.

The remainder of the work is organized as follows.
Section~\ref{sec:background} introduces the truncation model,
the qubit and qudit encodings,
and the embedded two-level $SU(2)$ primitives
used to express single-qudit implementations.
Section~\ref{sec:regime1} develops Regime~1 (product-formula) implementations:
Sections~\ref{sec:regime1_qubit} and~\ref{sec:regime1_qudit}
give the corresponding non-Clifford cost models for qubits and qudits, respectively,
Section~\ref{sec:regime1_asymptotics} studies the asymptotic synthesis behavior,
and Section~\ref{sec:regime1_crossover}
provides synthesis-parametrized crossover criteria.
Section~\ref{sec:regime2} develops Regime~2 (LCU/block-encoding) implementations:
Sections~\ref{sec:regime2_qubit} and~\ref{sec:regime2_qudit}
give the corresponding constructions and costs,
Section~\ref{sec:regime2_asymptotic} studies asymptotic behavior,
and Sections~\ref{sec:regime2_nocs_crossover} and~\ref{sec:regime2_codeswitch} 
analyze finite-$d$ encoding and hybrid finite-$d$ comparisons, respectively.
We discuss our findings in Section~\ref{sec:discussion}, and conclude in Section~\ref{sec:conclusions}.

\section{Problem Setup and Definitions}\label{sec:background}

A representative and widely occurring primitive in lattice Hamiltonian simulation
is onsite diagonal evolution generated by a quadratic function
of a local, digitized degree of freedom.
In this work, we use a digitized real scalar field operator $\phi_x$ at site $x$
as a concrete example,
focusing on the implementation of the single-site unitary $e^{-it\phi_x^2}$.
We encode $\phi_x$ using a uniform grid in the field-amplitude basis.
This digitization scheme is consistent with the standard formalism
for scalar fields in quantum simulation,
often referred to as the Jordan--Lee--Preskill (JLP) basis
~\cite{jordan_lee_preskill_scattering_2012,klco_savage_scalar_digitization_2018,Jordan:2011ci}.
Although our focus is strictly on the onsite mass term $\propto \phi_x^2$,
we note that the same ideas apply to terms $\propto \pi_x^2$,
since $\pi_x$ is diagonal in the Fourier-conjugate basis.
The qudit Fourier transform is Clifford~\cite{wang_qudits_review_2020},
so this basis change does not add non-Clifford cost
in the single-qudit setting;
by contrast, the $n_b$-qubit QFT is generally non-Clifford
and can add overhead in the qubit encoding.
Beyond our focus on real scalar fields---important in lattice $\phi^4$ theories
and dynamical-pion nuclear EFTs~\cite{watson2023quantumalgorithmssimulatingnuclear}---
operators of this form also arise in lattice gauge theories
~\cite{PhysRevA.73.022328,banuls_lgt_qtech_2020,Shaw:2020udc},
condensed matter and quantum optics
~\cite{2021PhRvL.126p3203C,PhysRevLett.81.3108},
and quantum chemistry~\cite{C9SC01313J}.

To truncate and discretize $\phi_x$,
we choose a maximum field value $\phi_{\max}$ and a grid spacing $\delta_\phi$.
Defining
\begin{equation}
    M = \frac{\phi_{\max}}{\delta_\phi} \in \mathbb{Z}, \qquad d = 2M+1,
\end{equation}
we label the basis states as $\{\vert n\rangle\}_{n=0}^{d-1}$.
Note that this symmetric grid implies odd $d$,
which we leave implicit throughout unless stated otherwise.
This bounded, uniformly spaced field-amplitude truncation is precisely the finite-dimensional local Hilbert-space digitization of the Jordan--Lee--Preskill (JLP) framework for quantum simulation of scalar field theories~\cite{jordan_lee_preskill_scattering_2012,Jordan:2011ci,klco_savage_scalar_digitization_2018}, to which the encoding used here directly corresponds.
The field operator on site $x$ is defined diagonally in this basis as:
\begin{equation}\label{eq:phi_operator}
    \phi_x = \sum_{n=0}^{d-1} \lambda_n\,\vert n\rangle\!\langle n\vert, \qquad
    \lambda_n = -\phi_{\max} + n\,\delta_\phi.
\end{equation}
A single $d$-level qudit naturally represents these $2M+1$ equally spaced values
\begin{equation}
    \{-\phi_{\max},\; -\phi_{\max} + \delta_\phi,\; \dots,\; +\phi_{\max}\}.
\end{equation}
For qubit implementations, we use an $n_b=\lceil \log_2 d\rceil$-qubit
register, but the precise embedding depends on the construction being
compared. In the product-formula setting of Sec.~\ref{sec:regime1}, we use the
standard binary embedding of the $d$ physical field-amplitude states
$\{\ket{n}\}_{n=0}^{d-1}$ into the first $d$ computational basis states
of the $2^{n_b}$-dimensional register. The digitized field operator is taken to act diagonally, with the remaining $2^{n_b}-d$ computational basis states lying outside the truncation and therefore unused.
In the LCU/block-encoding setting of Sec.~\ref{sec:regime2}, we instead use the
nearest signed-binary register supported by $n_b$ qubits, following the standard projector-LCU baseline of Refs.~\cite{Su_2021_PRXQ,Spagnoli_2025}. The qudit and qubit constructions are motivated by the same physical $d=2M+1$ target truncation: the qudit construction acts natively on the $d$-level field-amplitude truncation, while the qubit baseline implements a signed-binary extension of the rounded squared operator that agrees with the physical target on the physical interval $\ell\in\{-M,\ldots,M\}$. The Regime~2 finite-$d$ comparisons in this paper are therefore construction-specific comparisons between the native qudit LCU construction and the binary qubit projector-LCU.

To perform an arbitrary $SU(d)$ rotation on a single qudit,
one needs a universal basis of operations.
This basis can be constructed from a collection
of embedded $SU(2)$ rotations that couple different basis states.
The operators that generate these fundamental $SU(2)$ subalgebras
are defined as generalizations of the Pauli matrices
~\cite{Gustafson_2021}.
Let $\{|n\rangle\}_{n=0}^{d-1}$ be the computational basis.
For $0\le b<c\le d-1$,
define the embedded two-level generators
\begin{align}
X^{(b,c)} &= |b\rangle\!\langle c| + |c\rangle\!\langle b|,\\
Y^{(b,c)} &= -i\,|b\rangle\!\langle c| + i\,|c\rangle\!\langle b|,\\
Z^{(b,c)} &= |b\rangle\!\langle b| - |c\rangle\!\langle c|.
\end{align}
We can then define rotations
\begin{equation}
R_{G}^{(b,c)}(\theta) = \exp\!\left(-\tfrac{i\theta}{2}\,G^{(b,c)}\right), \qquad G\in\{X,Y,Z\},
\end{equation}
which act only on the $\{|b\rangle,|c\rangle\}$ subspace.

Throughout, unless stated otherwise,
we compare a single fault-tolerant $d$-level qudit
with a register of $n_b$ qubits encoding the field operator at site $x$.

Ideally, we would have a closed-form synthesis bound
for approximating the $d$-dimensional diagonal qudit gate $e^{-it\phi_x^2}$
to accuracy $\varepsilon$ using a fixed fault-tolerant qudit gate set.
For $\mathrm{SU}(p)$ with prime $p$,
a worst-case lower bound for Clifford+T synthesis is given in Ref.~\cite{PrakashKalraJain2021}.
However, although volume/covering arguments provide existential worst-case lower bounds
on $\varepsilon$-approximation in $SU(p)$ \cite{PrakashKalraJain2021},
these bounds do not certify that a particular structured diagonal,
such as $e^{-it\phi_x^2}$, is worst case.
Determining tight cost for a specific target typically requires finding
the nearest implementable exactly synthesizable unitary $V$
and the shortest word implementing it,
leading to NP-complete decision problems \cite{MANDERS1978168}.
Extending the single-qutrit ($d=3$) Clifford+$R$ synthesis analysis
of Ref.~\cite{gustafson2025synthesissinglequtritcircuits}
to arbitrary odd $d$ would require
(i) a dimension-$d$ normal form and exact-synthesis framework
for a fixed fault-tolerant gate set,
and (ii) dimension-specific number-theoretic/algorithmic analysis.
Accordingly, deriving a closed-form non-Clifford cost---or proving
that $e^{-it\phi_x^2}$ saturates generic lower bounds---for arbitrary $d$
is beyond our scope and is left for future work. 

Further, in the finite-$d$ compiler-target analysis in Sections~\ref{sec:regime1_crossover} and~\ref{sec:regime2_nocs_crossover}, we restrict to prime local dimensions.
This restriction reflects the current state of the fault-tolerant qudit-synthesis literature: Clifford+$T$ normal-form/synthesis results and qudit magic-state distillation constructions are most directly developed for prime-dimensional qudits~\cite{gottesman_qudit_ft_1998,hostens_dehaene_de_moor_qudit_stabilizer_2005,campbell_anwar_browne_msd_prime_2012,PrakashKalraJain2021,gustafson2025synthesissinglequtritcircuits}.

Therefore, for resource estimates,
we first express algorithmic costs in terms of continuous one-parameter rotations
(e.g., qubit $R_z$ rotations or embedded qudit $SU(2)$ rotations).
In a fault-tolerant setting,
these rotations must then be approximately synthesized
over a discrete logical gate set to precision $\varepsilon$.
Accordingly, the primary resource metric we compare throughout
is the resulting number of non-Clifford gates
(i.e., the non-Clifford primitives required by the discrete fault-tolerant gate set after synthesis),
and in Sec.~\ref{sec:regime2_codeswitch} we specifically compare T-counts.

\section{Regime 1: Time Evolution via Product Formulas (Trotterization)}\label{sec:regime1}

In this section, we compare the fault-tolerant cost of implementing a single product-formula step of the onsite unitary
$e^{-it \phi_x^2}$ in the qubit and qudit encodings.
We refer the reader to Ref.~\cite{ChildsSuTranWiebeZhu2021CommutatorScaling}
for background on product formulas.
For fixed formula order and target simulation parameters, the required step count $r$ is encoding independent, so the
relative resource comparison reduces to the non-Clifford cost per step at matched accuracy.

\subsection{The Qubit Implementation}\label{sec:regime1_qubit}

\begin{lemma}\label{prop:mass-trotter-qubit}
 A single Trotter step of $e^{-it \phi_x^2}$
 on a register of $n_b = \lceil \log_2 d \rceil$ qubits
 can be implemented using at most $\frac{n_b(n_b+1)}{2}$ $R_z$ gates.
\end{lemma}
\begin{proof}
 We exploit the binary structure of the register
 by writing the integer level index as $n=\sum_{m=0}^{n_b-1} 2^m q_m$
 with $q_m\in\{0,1\}$ \cite{watson2023quantumalgorithmssimulatingnuclear}.
 Mapping these bits to Pauli operators via $q_m=\frac{1-Z^{(m)}}{2}$,
 the field operator becomes
\begin{equation}
    \phi_x = P\mathbb{I} + Q\sum_{m=0}^{n_b-1}2^m Z^{(m)},
\end{equation}
with $P = -\phi_{\max} + \frac{\delta_\phi}{2}(2^{n_b}-1)$,
$Q = -\frac{\delta_\phi}{2}$.
Here $Z^{(m)}$ denotes the Pauli Z operator acting on the $m$th qubit of the register,
with $m=0$ corresponding to the least significant bit.
Squaring the operator yields
\begin{equation}
\begin{aligned}
\phi_x^2
&= \left(P\mathbb{I} + Q\sum_{m=0}^{n_b-1}2^m Z^{(m)}\right)^2 \\
&= P^2\mathbb{I}
+ 2PQ \sum_{m=0}^{n_b-1} 2^m Z^{(m)} \\
&\quad + Q^2 \sum_{m,m'=0}^{n_b-1} 2^{m+m'} Z^{(m)} Z^{(m')}.
\end{aligned}
\end{equation}
The time-evolution operator $e^{-it\phi_x^2}$ is therefore a product of commuting rotations:
\begin{equation}
\begin{aligned}
e^{-it\,\phi_x^2}
&=\exp\!\Bigl[-it\Bigl(P^2\mathbb{I}+2PQ\sum_{m=0}^{n_b-1}2^m Z^{(m)} \\
&\qquad\qquad\qquad\quad +Q^2\sum_{m,m'=0}^{n_b-1}2^{m+m'} Z^{(m)}Z^{(m')}\Bigr)\Bigr].
\end{aligned}
\end{equation}
This requires $n_b$ $R_z(\theta)$ rotations (from the linear terms)
and $\frac{n_b(n_b-1)}{2}$ two-qubit $ZZ$ rotations (from the quadratic terms).
Since $e^{-i\frac{\theta}{2}(Z \otimes Z)}$ decomposes into two CNOTs
and one $R_z(\theta)$ gate,
the total number of $R_z(\theta)$ gates is thus $\frac{n_b^2 + n_b}{2}$.
\end{proof}

\subsection{The Qudit Implementation}\label{sec:regime1_qudit}
In the qudit computational basis $\{|n\rangle\}_{n=0}^{d-1}$, the digitized field operator
$\phi_x$ is diagonal with eigenvalues $\lambda_n$ (as defined in Sec.~\ref{sec:background}).
It follows that
\begin{equation}
    \phi_x^2 = \sum_{n=0}^{d-1} \lambda_n^2 \, |n\rangle\langle n|
\end{equation}
and the corresponding onsite evolution is the diagonal unitary
\begin{equation}
    e^{-it\phi_x^2}=\sum_{n=0}^{d-1} e^{-it \lambda_n^2} \, |n\rangle\langle n|\,.
\end{equation}

\begin{lemma}\label{prop:su2_irred}
 Let $U=\sum_{n=0}^{d-1} e^{-i\beta_n}|n\rangle\langle n|$ be a diagonal unitary on a single $d$-level qudit.
 Then $U$ can be implemented \emph{up to a global phase} using at most $d-1$
 $R_{Z}^{(k,k+1)}\!\left(\theta_k\right)$ rotations.
 In particular, taking $\beta_n=t\lambda_n^2$ yields a decomposition of a single Trotter step $e^{-it\phi_x^2}$.
\end{lemma}

\begin{proof}
We first consider an arbitrary diagonal unitary $U=\sum_{n=0}^{d-1} e^{-i\beta_n}|n\rangle\langle n|$
and show how to implement it up to a global phase using adjacent $R_Z^{(k,k+1)}$ rotations; we then specialize to $\beta_n=t\lambda_n^2$.
First, each adjacent $Z$-Givens rotation $R_Z^{(k,k+1)}(\theta_k)$ acts as
\begin{equation}
R_Z^{(k,k+1)}(\theta_k)= e^{-i\theta_k/2}|k\rangle\langle k| + e^{+i\theta_k/2}|k+1\rangle\langle k+1|,
\end{equation}
where $k = 0,\dots,d-2$.
All such rotations are diagonal and therefore commute, so on level $|n\rangle$ of the qudit the net phase is $e^{-i \alpha_n}$, where
\begin{equation}\label{eq:su2irred_linearsystem}
\alpha_n \;=\; \tfrac{1}{2}\theta_n - \tfrac{1}{2}\theta_{n-1},
\qquad
\theta_{-1}=\theta_{d-1}=0.
\end{equation}
Since $\sum_{n=0}^{d-1}\alpha_n=\tfrac{1}{2}(-\theta_{-1}+\theta_{d-1})=0$, the product lies in $SU(d)$.
The diagonal phases we wish to implement, $(\beta_n)_{n=0}^{d-1}$,
need not satisfy $\sum_n \beta_n = 0$,
but they can be realized up to a global phase by writing
$e^{-i\beta_n}=e^{-i\bar\beta}\,e^{-i(\beta_n-\bar\beta)}$ with
\begin{equation}
\bar\beta \;=\; \frac{1}{d}\sum_{n=0}^{d-1}\beta_n.
\end{equation}
In our application to $e^{-it\phi_x^2}$ we have $\beta_n=t\lambda_n^2$.
Thus we choose the rotation angles as
\begin{equation}
\theta_k \;=\; 2\sum_{n=0}^{k}\!\bigl(\beta_n-\bar\beta\bigr)
\quad (\mathrm{mod}\ 4\pi),
\end{equation}
so that $\alpha_n=\beta_n-\bar\beta$, i.e., the implemented phase equals $e^{-i(\beta_n-\bar\beta)}$. Hence
\begin{equation}
\prod_{k=0}^{d-2} R_{Z}^{(k,k+1)}\!\left(\theta_k\right)
\end{equation}
implements $U$ up to a global phase, so one such decomposition with $d-1$ Givens rotations exists.
Specializing to $\beta_n=t\lambda_n^2$ yields $e^{-it\phi_x^2}$ up to a global phase.
Moreover, for any target phases $(\alpha_n)_{n=0}^{d-1}$ with $\sum_n \alpha_n = 0$,
the linear system in Eq.~\ref{eq:su2irred_linearsystem}
has the unique solution
\begin{equation}
\theta_k = 2\sum_{n=0}^{k}\alpha_n.
\end{equation}
Thus the angles $(\theta_k)$ that realize a given diagonal (up to global phase) are unique modulo $4\pi$.
\end{proof}

Note that under the symmetric uniform truncation with odd $d$
(see Eq.~\ref{eq:phi_operator}),
the centered partial sums admit the closed form
\begin{equation}
\sum_{n=0}^k(\lambda_n^2-\mu)
=
\phi_{\max}^2
\frac{4(k+1)}{3(d-1)^2}\Bigl(k-\frac{d-2}{2}\Bigr)\bigl(k-(d-1)\bigr),
\end{equation}
where $\mu=\frac{1}{d}\sum_{n=0}^{d-1}\lambda_n^2$.
Since $k\in\{0,\dots,d-2\}$ and $(d-2)/2$ is a half-integer for odd $d$,
this expression never vanishes.
Hence $\theta_k \not\equiv 0\ (\mathrm{mod}\ 4\pi)$ for generic $t$,
and therefore $\theta_k\neq 0$.

\subsection{Asymptotic Scaling Comparison}\label{sec:regime1_asymptotics}
In this section, we compare the non-Clifford cost of the qudit and qubit implementations.
On the qubit side, we consider a fixed fault-tolerant qubit gate set
and define $C^{(\mathrm{qb})}_{R_z}(\delta)$
as the non-Clifford synthesis cost of a single-qubit $R_z$ rotation to accuracy $\delta$.
On the qudit side, we likewise fix a (dimension-$d$) fault-tolerant qudit gate set
and let $C^{\mathrm{(qd)}}_{SU(2)}(\delta)$ denote the non-Clifford synthesis cost
to approximate a single embedded two-level one-parameter $SU(2)$ rotation
to operator-norm error at most $\delta$.
As is standard in fault-tolerant settings,
neither the qubit $R_z$ rotations nor the qudit embedded $SU(2)$ rotations
are treated as native primitives:
in both encodings, these are continuous one-parameter targets
that must be compiled into a finite discrete gate set,
and our resource metric counts only the resulting non-Clifford gates.

Lemmas~\ref{prop:mass-trotter-qubit} and~\ref{prop:su2_irred}
reduce the specific Regime~1 constructions 
to circuits containing a known number of one-parameter rotations.
This allows a gate-set--agnostic asymptotic comparison
between the qubit implementation of Lemma~\ref{prop:mass-trotter-qubit}
and the qudit implementation obtained from Lemma~\ref{prop:su2_irred}
followed by independent synthesis of the resulting embedded two-level
$SU(2)$ rotations under a uniform per-rotation error allocation.
Rather than asserting an explicit $C^{\mathrm{(qd)}}_{SU(2)}$,
we ask what relative synthesis performance of those qudit primitives
would be required for the qudit construction
to use fewer non-Cliffords than the qubit baseline
at the same overall precision.
Theorem~\ref{prop:regime1_comparison} states the resulting requirement.

\begin{theorem}\label{prop:regime1_comparison}
Let $U=e^{-it\phi_x^2}$ be a single-site Trotter step on a $d$-level qudit, and let $n_b=\lceil\log_2 d\rceil$.
Consider the qudit implementation obtained by Lemma~\ref{prop:su2_irred},
with each embedded two-level $SU(2)$ rotation synthesized independently under a uniform per-rotation error allocation.
If this qudit decomposition is to asymptotically use fewer non-Cliffords than the qubit implementation
of Lemma~\ref{prop:mass-trotter-qubit}, then the non-Clifford synthesis cost of the required one-parameter $SU(2)$ rotations
must be smaller than that of qubit-optimal $R_z$ synthesis in Clifford+T by a factor
\begin{equation}
\Theta\!\left(\frac{(\log d)^2}{d}\right)
=\Theta\!\left(\frac{n_b^2}{2^{n_b}}\right),
\end{equation}
i.e., by an exponentially strong advantage in $n_b$.
\end{theorem}

\begin{proof}
By Lemma~\ref{prop:su2_irred}, $U=e^{-it\phi_x^2}$ uses $(d-1)$ nontrivial one-parameter embedded $SU(2)$ rotations; denote this count by
\begin{equation}
L^{(1)}_{\mathrm{qd}}(d) = d-1.
\end{equation}
Assume each such $SU(2)$ rotation is approximated using a fixed discrete gate set,
with synthesis cost $C^{\mathrm{(qd)}}_{SU(2)}(\delta)$
to achieve operator-norm error at most $\delta$.
Choosing
\begin{equation}
\delta \;=\; \frac{\varepsilon}{L^{(1)}_{\mathrm{qd}}(d)}
\end{equation}
ensures total error at most $\varepsilon$ by a standard triangle-inequality bound, so the qudit non-Clifford count satisfies
\begin{equation}
N^{\mathrm{PF}}_{\mathrm{qudit}}(d,\varepsilon)
\;\lesssim\;
L^{(1)}_{\mathrm{qd}}(d)\,C^{\mathrm{(qd)}}_{SU(2)}\!\left(\frac{\varepsilon}{L^{(1)}_{\mathrm{qd}}(d)}\right).
\end{equation}
By Lemma~\ref{prop:mass-trotter-qubit}, $U$ can be implemented using
\begin{equation}
L^{(1)}_{\mathrm{qb}}(d)
\;=\;
\frac{n_b^2+n_b}{2}
\end{equation}
single-qubit $R_z$ rotations on $n_b$ qubits. Allocating per-rotation accuracy
\begin{equation}
\delta_{\mathrm{qb}}
\;=\;
\frac{\varepsilon}{L^{(1)}_{\mathrm{qb}}(d)},
\end{equation}
the corresponding qubit non-Clifford count is
\begin{equation}
N^{\mathrm{PF}}_{\mathrm{qubit}}(d,\varepsilon)
\;=\;
L^{(1)}_{\mathrm{qb}}(d)\,
C^{(\mathrm{qb})}_{R_z}\!\left(\frac{\varepsilon}{L^{(1)}_{\mathrm{qb}}(d)}\right).
\end{equation}
Under the per-rotation accuracy model, a necessary condition for the qudit implementation of Lemma~\ref{prop:su2_irred} to use fewer non-Cliffords than the qubit implementation is
\begin{equation}
N^{\mathrm{PF}}_{\mathrm{qudit}}(d,\varepsilon)
\;\lesssim\;
N^{\mathrm{PF}}_{\mathrm{qubit}}(d,\varepsilon),
\end{equation}
i.e.,
\begin{equation}
L^{(1)}_{\mathrm{qd}}(d)\,C^{\mathrm{(qd)}}_{SU(2)}\!\left(\frac{\varepsilon}{L^{(1)}_{\mathrm{qd}}(d)}\right)
\;\lesssim\;
L^{(1)}_{\mathrm{qb}}(d)\,
C^{(\mathrm{qb})}_{R_z}\!\left(\frac{\varepsilon}{L^{(1)}_{\mathrm{qb}}(d)}\right).
\end{equation}
Rearranging yields
\begin{equation}\label{eq:su2_inequality}
C^{\mathrm{(qd)}}_{SU(2)}\!\left(\frac{\varepsilon}{L^{(1)}_{\mathrm{qd}}(d)}\right)
\;\lesssim\;
\frac{L^{(1)}_{\mathrm{qb}}(d)}{L^{(1)}_{\mathrm{qd}}(d)}\,
C^{(\mathrm{qb})}_{R_z}\!\left(\frac{\varepsilon}{L^{(1)}_{\mathrm{qb}}(d)}\right).
\end{equation}
Since $L^{(1)}_{\mathrm{qb}}(d)=\Theta((\log d)^2)$, the prefactor satisfies
\begin{equation}
\frac{L^{(1)}_{\mathrm{qb}}(d)}{d-1}
\;=\;
\Theta\!\left(\frac{(\log d)^2}{d}\right)
\;=\;
\Theta\!\left(\frac{n_b^2}{2^{n_b}}\right),
\end{equation}
which is exponentially small in $n_b=\lceil\log_2 d\rceil$.
Thus, for fixed $\varepsilon$, the Lemma~\ref{prop:su2_irred} decomposition plus independent-synthesis route can asymptotically beat the qubit implementation only if embedded $SU(2)$ rotations admit exponentially cheaper synthesis (as a function of $n_b$) than qubit $R_z$ rotations.
\end{proof}

For any fixed finite discrete gate set, approximating a continuous one-parameter subgroup to precision $\delta$ requires
$C^{\mathrm{(qd)}}_{SU(2)}(\delta)=\Omega(\log(1/\delta))$
by standard volume-counting arguments \cite{10.1063/1.1495899}.
Moreover, for common fault-tolerant gate sets admitting number-theoretic synthesis,
known procedures achieve $C_{\mathrm{prim}}(\delta)=O(\log(1/\delta))$
\cite{gustafson2025synthesissinglequtritcircuits, ross_selinger_optimal_rz_2014},
so there is no asymptotic room for an exponentially improving dependence
on $n_b=\lceil\log_2 d\rceil$.
Consequently, we view the exponential per-rotation advantage demanded
by Theorem~\ref{prop:regime1_comparison}
as incompatible with the standard fault-tolerant gate-synthesis landscape
and therefore do not expect the qudit encoding to asymptotically use fewer non-Cliffords than the qubit construction.

\subsection{Finite-d Crossover Analysis}\label{sec:regime1_crossover}

The $SU(2)$ decomposition still provides a meaningful comparison
at the level of the number of continuous one-parameter rotations required.
For example, in qutrit Clifford+$R$ synthesis,
Ref.~\cite{gustafson2025synthesissinglequtritcircuits}
focuses on approximating diagonal two-level rotations,
since these serve as the basic building blocks for $SU(3)$ synthesis.
Furthermore, as discussed above,
for any fixed finite discrete gate set,
approximating a continuous one-parameter subgroup to precision $\delta$ requires
$\Omega(\log(1/\delta))$ non-Cliffords
by standard volume-counting arguments \cite{10.1063/1.1495899}.
Therefore, both qubit rotations and embedded-qudit $SU(2)$ primitives
admit logarithmic-in-precision synthesis, $\Omega(\log(1/\delta))$,
and we can study the synthesis behavior required
of the fault-tolerant qudit implementation as the truncation $d$ grows.

Assume we wish to implement the onsite unitary to accuracy $\varepsilon$.
Starting from Eq.~\ref{eq:su2_inequality},
we model the cost of a fault-tolerant continuous one-parameter embedded $SU(2)$ rotation as
\begin{equation}
C^{\mathrm{(qd)}}_{SU(2)}\approx a\,\log_2(1/\delta)
\end{equation}
where $a>0$ is a constant prefactor that captures gate-set and compilation overheads.
For the qubit baseline,
we use the single-qubit $R_z$ synthesis model \cite{Kliuchnikov2023shorterquantum},
\begin{equation}
C^{(\mathrm{qb})}_{R_z}(\delta)=0.57\log_2(1/\delta)+8.83.
\end{equation}
Substituting $\delta=\varepsilon/L^{(1)}_{\mathrm{qd}}(d)$ for the qudit decomposition
and $\delta=\varepsilon/L^{(1)}_{\mathrm{qb}}(d)$ for the qubit construction
yields the approximate break-even condition
\begin{equation}
\begin{split}
L^{(1)}_{\mathrm{qd}}(d)\,a\,\log_2\!\left(\frac{L^{(1)}_{\mathrm{qd}}(d)}{\varepsilon}\right)
\\
\lesssim
L^{(1)}_{\mathrm{qb}}(d)\left[
0.57\log_2\!\left(\frac{L^{(1)}_{\mathrm{qb}}(d)}{\varepsilon}\right)+8.83
\right].
\end{split}
\end{equation}
Rearranging gives an explicit requirement on the qudit synthesis prefactor,
\begin{equation}
a \lesssim a^{\mathrm{PF}}_{\max}(d,\varepsilon)
=
\frac{
L^{(1)}_{\mathrm{qb}}(d)\left[0.57\log_2\!\left(L^{(1)}_{\mathrm{qb}}(d)/\varepsilon\right)+8.83\right]
}{
L^{(1)}_{\mathrm{qd}}(d)\log_2\!\left(L^{(1)}_{\mathrm{qd}}(d)/\varepsilon\right)
}.
\end{equation}

For $\varepsilon=10^{-6}$
(a representative value),
and restricting to prime local dimensions in the fixed-encoding comparison,
this yields
\begin{equation}
\begin{split}
a^{\mathrm{PF}}_{\max}(3,10^{-6}) \approx 1.51,\\
a^{\mathrm{PF}}_{\max}(5,10^{-6}) \approx 1.48,\\
a^{\mathrm{PF}}_{\max}(7,10^{-6}) \approx 0.96.
\end{split}
\end{equation}
The quantity $a^{\mathrm{PF}}_{\max}(d,\varepsilon)$ should be read as a compiler target:
if an embedded two-level qudit $SU(2)$ rotation can be synthesized with cost
$a<a^{\mathrm{PF}}_{\max}(d,\varepsilon)$,
then the qudit implementation uses fewer non-Clifford gates than the qubit baseline.

It is also useful to compare this compiler target with the effective prefactor that would reproduce the qubit
$R_z$ synthesis cost at the same primitive precision required of the qudit rotations,
\begin{equation}
a^{\mathrm{PF}}_{R_z}(d,\varepsilon)
=
\frac{
0.57\log_2\!\left(L^{(1)}_{\mathrm{qd}}(d)/\varepsilon\right)+8.83
}{
\log_2\!\left(L^{(1)}_{\mathrm{qd}}(d)/\varepsilon\right)
}.
\end{equation}
When
$a^{\mathrm{PF}}_{\max}(d,\varepsilon)>a^{\mathrm{PF}}_{R_z}(d,\varepsilon)$,
the qudit construction can tolerate embedded-$SU(2)$ synthesis that is worse than qubit
$R_z$ synthesis at the same primitive precision.
When
$a^{\mathrm{PF}}_{\max}(d,\varepsilon)<a^{\mathrm{PF}}_{R_z}(d,\varepsilon)$,
a genuine per-primitive qudit synthesis advantage is required.

For $\varepsilon=10^{-6}$,
we find $a^{\mathrm{PF}}_{\max}(d,\varepsilon)>a^{\mathrm{PF}}_{R_z}(d,\varepsilon)$
only for $d=3$ and $d=5$ among the computed prime dimensions.
Determining whether these compiler targets are attainable for $d\geq 3$
requires explicit synthesis estimates for the relevant embedded primitives.

\section{Regime 2: Time Evolution via Linear Combination of Unitaries (LCU)}\label{sec:regime2}
In the block-encoding/LCU framework,
the motivation for using qudits is stronger than in the continuous-rotation setting:
when we express the same diagonal operator as a linear combination of unitaries,
some terms that are non-Clifford on qubits can become Clifford on qudits.
Prominent examples include the $d$-level incrementer and the quantum Fourier transform,
both of which are Clifford operations on a single qudit
\cite{wang_qudits_review_2020}.
Our goal is to determine whether, at fixed physical precision,
increasing the local dimension
(i.e., using a $d$-level qudit instead of $\lceil\log_2 d\rceil$ qubits)
can yield a constant-factor or asymptotic reduction in non-Clifford resources.

We consider two comparisons within Regime~2: a fully fixed-encoding qudit-versus-qubit comparison, and a secondary idealized hybrid code-switching comparison. In both comparisons, the qubit baseline is the standard nearest-register signed-binary projector-LCU construction of Refs.~\cite{Su_2021_PRXQ,Spagnoli_2025} on $n_b=\lceil\log_2 d\rceil$ qubits, which implements the signed-binary extension of the rounded squared operator that agrees with the physical $d=2M+1$ target on $\ell\in\{-M,\ldots,M\}$. The main distinction between the two comparisons is that the fixed-encoding model asks what native
fault-tolerant qudit synthesis performance would be required for an
advantage, while the hybrid model obtains an absolute qubit
Clifford+$T$ comparison by allowing non-Clifford subroutines to be
implemented after idealized qubit--qudit code-switching. The finite-$d$ Regime~2 results below are therefore construction-specific comparisons between the native qudit LCU construction and this standard qubit baseline, and not optimal lower bounds over all possible qubit encodings of the truncation.

In the fixed-encoding comparison, the qudit implementation remains entirely in the qudit encoding, and we derive the corresponding break-even condition
$a \lesssim a_{\max}^{\mathrm{LCU}}(d,t,\varepsilon_{\mathrm{sim}})$.
In the hybrid comparison, an \(n_b=\lceil \log_2 d\rceil\)-qubit index register is
allowed to code-switch to a single logical \(d\)-level qudit and back, giving an
idealized absolute \(T\)-count comparison. Here, code-switching is modeled as an
isometry from the valid index subspace to \(\mathbb C^d\), followed by the
inverse isometry after the qudit operation. The invalid computational states
\(|d\rangle,\ldots,|2^{n_b}-1\rangle\) are assumed to be unpopulated in the
ideal circuit; the physical cost of implementing the isometry is not included,
except through the overhead-budget analysis in
Sec.~\ref{sec:regime2_codeswitch}.
For the purposes of this section, we focus on the hybrid qudit encoding, with details of the fixed-encoding implementation given in Appendix~\ref{app:reg2-crossover}.

Given a Hamiltonian $H_i$ expressed as a linear combination of unitaries,
\begin{equation}
    H_i = \sum_{j=0}^{M_i-1} h_{ij} U_{ij}, \qquad
    \lambda_i = \sum_{j=0}^{M_i-1} |h_{ij}|,
\end{equation}
one can define a preparation unitary $\text{PREP}_i$ and a selection unitary $\text{SELECT}_i$ by
\begin{align}
    \text{PREP}_i &: |0\rangle \mapsto
    \sum_{j=0}^{M_i-1} \sqrt{\frac{|h_{ij}|}{\lambda_i}}\,|j\rangle, \\
    \text{SELECT}_i &= \sum_{j=0}^{M_i-1} |j\rangle\!\langle j| \otimes U_{ij}.
\end{align}
Here the first register (prepared by $\text{PREP}_i$)
is the term-selection register, which we call the index register;
$\text{SELECT}_i$ acts on both this register and on the system register.
Projecting the index register onto $|0\rangle$ then yields a block-encoding of $H_i$
\cite{hardy2024optimizedquantumsimulationalgorithms, childs_wiebe_lcu_2012}:
\begin{equation}
    \langle 0|\,\text{PREP}_i^\dagger\,\text{SELECT}_i\,\text{PREP}_i\,|0\rangle
    = \frac{H_i}{\lambda_i}.
\end{equation}
This block-encoding can then be used as input
to standard Hamiltonian-simulation primitives
(e.g., qubitization/quantum signal processing
and its generalization via quantum singular value transformation)
to implement real-time evolution $e^{-i H_i t}$,
and more general functions of $H_i$,
with polylogarithmic dependence on target precision.
For a fault-tolerant resource comparison in Regime~2,
we focus on the dominant oracle costs of $\text{PREP}_i$ and $\text{SELECT}_i$,
since these set the query complexity and, in our constructions,
drive the leading non-Clifford gate counts.
For background on block-encodings and their use in Hamiltonian simulation,
see Refs.~\cite{childs_wiebe_lcu_2012,low_chuang_qubitization_2019,gilyen_qsvt_blockencoding_2019}.

\subsection{The Qubit Baseline}\label{sec:regime2_qubit}
The qubit baseline studied here is the standard nearest-register signed-binary projector-LCU construction of Refs.~\cite{Su_2021_PRXQ,Spagnoli_2025}. The physical target is the same $d=2M+1$ field-amplitude truncation as in the qudit construction; the qubit baseline implements the nearest signed-binary extension of the rounded squared operator on
\begin{equation}
n_b = \lceil \log_2 d\rceil
\end{equation}
qubits. Starting from the uniform field operator in Eq.~\ref{eq:phi_operator}, we define the centered field label
\begin{equation}
\ell = n - M.
\end{equation}
We label a computational basis state of the $n_b$-qubit register by a sign bit $\ell_{\mathrm{sign}}\in\{0,1\}$ and magnitude bits $\ell_{n_b-2},\ldots,\ell_0\in\{0,1\}$,
\begin{equation}
|\ell\rangle = |\ell_{\mathrm{sign}},\ell_{n_b-2},\ldots,\ell_0\rangle,
\end{equation}
and assign the integer field label
\begin{equation}
\ell = (-1)^{\ell_{\mathrm{sign}}}\sum_{r=0}^{n_b-2}2^r \ell_r,
\qquad
|\ell| = \sum_{r=0}^{n_b-2} 2^r \ell_r,
\end{equation}
so that
\begin{equation}
\ell \in \{-(2^{n_b-1}-1),\ldots,2^{n_b-1}-1\}.
\end{equation}
Note that the two computational strings with all magnitude bits zero (sign bit $\ell_{\mathrm{sign}}=0$ or $\ell_{\mathrm{sign}}=1$) are both assigned the field label $\ell=0$. The redundant negative-zero string is therefore not excluded; it is a second computational basis representative of $\ell=0$. Because the squared operator below depends only on the magnitude bits $(\ell_r)$, the sign bit is irrelevant whenever the magnitude is zero, and the negative-zero string is an invariant zero-eigenvalue state of the squared operator.

Rather than restricting to the physical interval $\{-M,\ldots,M\}$ and enforcing a physical-subspace projector, we work directly with the full $n_b$-qubit register. With a slight abuse of notation, we denote by $\phi_x^2$ the corresponding rounded diagonal operator defined on every computational basis state of the register,
\begin{equation}
\phi_x^2
=
\delta_\phi^2
\sum_{\ell_{\mathrm{sign}},\,\ell_{n_b-2},\ldots,\ell_0\in\{0,1\}}
\!\!\!\!\!\!
\ell^2\,
|\ell\rangle\langle \ell|,
\end{equation}
where $\ell$ on the right-hand side is the integer label assigned above. By construction, this operator agrees with the physical target on the physical interval $\ell\in\{-M,\ldots,M\}$ and is the standard signed-binary extension to the full register used throughout Refs.~\cite{Su_2021_PRXQ,Spagnoli_2025}. Equivalently, indexing by the integer label,
\begin{equation}
\phi_x^2
=
\delta_\phi^2
\sum_{\ell=-(2^{n_b-1}-1)}^{2^{n_b-1}-1}
\ell^2 |\ell\rangle\langle \ell|,
\end{equation}
with the understanding that $\ell=0$ is realized by both sign-bit values.
Since the sign bit drops out upon squaring, we have
\begin{equation}
\ell^2
=
\left(\sum_{r=0}^{n_b-2}2^r \ell_r\right)^2
=
\sum_{r,s=0}^{n_b-2}2^{r+s}\ell_r\ell_s.
\end{equation}
Accordingly, for each pair $(r,s)$ we define the diagonal bit-projector on the full $n_b$-qubit register,
\begin{equation}
P_{r,s}
=
\!\!
\sum_{\ell_{\mathrm{sign}},\ell_{n_b-2},\ldots,\ell_0\in\{0,1\}}
\!\!\!\!\!\!\!\!\!
\ell_r\ell_s\,|\ell\rangle\langle \ell|,
\end{equation}
which depends only on the magnitude bits and is therefore insensitive to the negative-zero redundancy. Because each bit $\ell_r$ is either $0$ or $1$, the product $\ell_r\ell_s$ is equal to $1$ exactly when the $r$th and $s$th non-sign bits are both equal to $1$, and is otherwise $0$. Hence
\begin{equation}
P_{r,s}|\ell\rangle
=
\begin{cases}
|\ell\rangle, & \ell_r=\ell_s=1,\\
0, & \text{otherwise},
\end{cases}
\end{equation}
so $P_{r,s}$ is the projector onto the span of basis states whose selected non-sign bits are both equal to $1$. Substituting the bit expansion of $\ell^2$ into the diagonal operator gives
\begin{equation}\label{eq:phi_qubit_expansion}
\phi_x^2
=
\delta_\phi^2
\sum_{r,s=0}^{n_b-2}2^{r+s}P_{r,s}.
\end{equation}
Since the operators $P_{r,s}$ are projectors rather than unitaries, we give a block encoding of $\phi_x^2$ using the projector-based LCU methods of Refs.~\cite{Su_2021_PRXQ,Spagnoli_2025}. The PREP unitary acts on the index and ancilla registers and does not move the system register out of any computational-basis subspace, while the SELECT unitary is diagonal in the signed-binary computational basis (it conditions on the index pair $(r,s)$ and applies a bit-test phase on the system); both are therefore well-defined on every computational basis state of the $n_b$-qubit Hilbert space.

\begin{lemma}\label{lem:qubit_lcu_phi2}
The operator $\phi_x^2$ admits a block encoding with normalization factor
\begin{equation}
\alpha = \delta_\phi^2\left(2^{n_b-1}-1\right)^2
\end{equation}
on an $n_b$-qubit system register, using two index registers of $\lceil \log_2(n_b-1)\rceil$ qubits each and one additional ancilla qubit. Equivalently, there exist unitaries $\mathrm{PREP}$ and $\mathrm{SELECT}$ such that
\begin{equation}
\bigl(\langle 0| \otimes I\bigr)\,
\mathrm{PREP}^\dagger \,\mathrm{SELECT}\,\mathrm{PREP}\,
\bigl(|0\rangle \otimes I\bigr)
=
\frac{\phi_x^2}{\alpha},
\end{equation}
where $|0\rangle$ denotes the joint initial state of the two index registers and the ancilla qubit. Moreover, the $T$ count for one block-encoding call is
\begin{equation}
T_{\phi_x^2}^{\mathrm{LCU}}(\varepsilon)
=
32b_r+24n_b-116,
\end{equation}
where
\begin{equation}
b_r
=
\left\lceil
\frac{1}{2}
\log_2\!\left(\frac{9\pi^2}{2\varepsilon}\right)
\right\rceil.
\end{equation}
\end{lemma}

\begin{proof}
Following Refs.~\cite{Su_2021_PRXQ,Spagnoli_2025}, we implement the projectors $P_{r,s}$ within a block encoding by introducing a single ancilla qubit $b$ prepared in the state
\begin{equation}
|+\rangle_b = \frac{|0\rangle_b + |1\rangle_b}{\sqrt{2}},
\end{equation}
together with the weighted index-register state
\begin{equation}
|\psi\rangle_{rs}
=
\frac{1}{\sqrt{\Lambda}}
\sum_{r,s=0}^{n_b-2} 2^{(r+s)/2} |r,s\rangle,
\end{equation}
where
\begin{equation}
\Lambda
=
\sum_{r,s=0}^{n_b-2} 2^{r+s}.
\end{equation}
Let $\mathrm{PREP}$ be the unitary defined by
\begin{equation}
\mathrm{PREP}\,|0\rangle_{rs}|0\rangle_b
=
|\psi\rangle_{rs}|+\rangle_b.
\end{equation}

For each pair $(r,s)$, define a unitary $\mathrm{SELECT}_{r,s}$ acting on the ancilla and signed-binary field register by
\begin{equation}
\mathrm{SELECT}_{r,s}\,
|b\rangle|\ell\rangle
=
(-1)^{\,b\,(\ell_r\ell_s\oplus 1)}
|b\rangle|\ell\rangle.
\end{equation}
Equivalently, conditioned on the chosen pair $(r,s)$, the circuit queries the $r$th and $s$th non-sign bits of the signed-binary label and applies a $Z$ to the ancilla unless both selected bits are equal to $1$. Therefore, for every basis state $|\ell\rangle$,
\begin{equation}
(\langle +|_b \otimes I)\,
\mathrm{SELECT}_{r,s}\,
(|+\rangle_b \otimes |\ell\rangle)
=
\ell_r \ell_s\,|\ell\rangle,
\end{equation}
and hence
\begin{equation}
(\langle +|_b \otimes I)\,
\mathrm{SELECT}_{r,s}\,
(|+\rangle_b \otimes I)
=
P_{r,s}.
\end{equation}
Thus the desired projector appears as the $|+\rangle_b$-to-$|+\rangle_b$ matrix element of a unitary.

The normalization is
\begin{equation}
\Lambda
=
\sum_{r,s=0}^{n_b-2} 2^{r+s}
=
\left(\sum_{r=0}^{n_b-2} 2^r\right)^2
=
\left(2^{n_b-1}-1\right)^2.
\end{equation}
Next define the full selection unitary
\begin{equation}
\mathrm{SELECT}
=
\sum_{r,s=0}^{n_b-2}
|r,s\rangle\langle r,s|
\otimes
\mathrm{SELECT}_{r,s},
\end{equation}
which applies the corresponding bit-test unitary conditioned on the index pair $(r,s)$. Sandwiching this operator between the weighted index state and the ancilla state gives
\begin{equation}
\begin{split}
(\langle\psi|_{rs}\langle +|_b \otimes I)\,\\
\mathrm{SELECT}\,
(|\psi\rangle_{rs}|+\rangle_b \otimes I)
=
\frac{1}{\Lambda}
\sum_{r,s=0}^{n_b-2}
2^{r+s} P_{r,s}.
\end{split}
\end{equation}
Using the projector decomposition above,
\begin{equation}
\phi_x^2
=
\delta_\phi^2
\sum_{r,s=0}^{n_b-2} 2^{r+s} P_{r,s},
\end{equation}
we obtain
\begin{equation}
(\langle\psi|_{rs}\langle +|_b \otimes I)\,
\mathrm{SELECT}\,
(|\psi\rangle_{rs}|+\rangle_b \otimes I)
=
\frac{\phi_x^2}{\delta_\phi^2 \Lambda}.
\end{equation}
Therefore this construction gives an exact qubit block encoding of $\phi_x^2$ with normalization factor
\begin{equation}
\alpha
=
\delta_\phi^2 \Lambda
=
\delta_\phi^2 \left(2^{n_b-1}-1\right)^2.
\end{equation}
Equivalently,
\begin{equation}
\bigl(\langle 0|_{rs}\langle 0|_b \otimes I\bigr)\,
\mathrm{PREP}^\dagger
\mathrm{SELECT}
\mathrm{PREP}
\bigl(|0\rangle_{rs}|0\rangle_b \otimes I\bigr)
=
\frac{\phi_x^2}{\alpha},
\end{equation}
which is the claimed exact block encoding at the oracle level.

For the fault-tolerant cost, we specialize the kinetic-term construction of Refs.~\cite{Su_2021_PRXQ,Spagnoli_2025} to a single signed-binary field register. One $\mathrm{PREP}$ oracle for the weighted bit-pair state has Toffoli count
\begin{equation}
T_{\mathrm{prep}}^{(\mathrm{Toff})}
=
4 b_r + 2 n_b - 16,
\qquad
b_r
=
\left\lceil
\frac{1}{2}
\log_2\!\left(\frac{9\pi^2}{2\varepsilon}\right)
\right\rceil,
\end{equation}
and $\mathrm{PREP}^\dagger$ has the same cost. The selector consists only of the two bit queries and the final controlled phase, giving
\begin{equation}
T_{\mathrm{sel}}^{(\mathrm{Toff})}
=
2(n_b-1),
\qquad
T_{\mathrm{sel}}^{(T)}
=
20.
\end{equation}
Using the convention $1$ Toffoli $= 4T$~\cite{PhysRevA.87.022328}
, the full block-encoding call
\begin{equation}
W
=
(\mathrm{PREP}^\dagger \otimes I)\,
\mathrm{SELECT}\,
(\mathrm{PREP} \otimes I)
\end{equation}
therefore has total $T$ count
\begin{align}
T_{\phi_x^2}^{\mathrm{LCU}}(\varepsilon)
&=
4\Bigl(
T_{\mathrm{prep}}^{(\mathrm{Toff})}
+
T_{\mathrm{prep}}^{(\mathrm{Toff})}
+
T_{\mathrm{sel}}^{(\mathrm{Toff})}
\Bigr)
+
T_{\mathrm{sel}}^{(T)}
\\
&=
32 b_r + 24 n_b - 116.
\end{align}
\end{proof}
\subsection{The Qudit Implementation}\label{sec:regime2_qudit}
For the qudit case, we first define the generalized Pauli \(Z\) on a \(d\)-level qudit
\cite{wang_qudits_review_2020}
\begin{equation}
    Z_d \;=\; \sum_{j=0}^{d-1} \omega^{\,j}\,|j\rangle\!\langle j|, 
\end{equation}
where \(\omega = e^{2\pi i/d}\), which belongs to the qudit Clifford group.

\subsubsection{Oracle Definitions}\label{sec:regime2_qudit_background}

\begin{lemma}\label{prop:qudit-pauli-expansion}
The squared field operator $\phi_x^2$ admits a unique expansion in the generalized Pauli basis as
\begin{equation}
\phi_x^2 = \sum_{r=0}^{d-1} \beta_r Z_d^r,
\end{equation}
with
\begin{equation}\label{eq:coeff_beta_decomp}
\begin{aligned}
\beta_0 &= \phi_{\max}^2\,\frac{d+1}{3(d-1)}, \\
\beta_r &= \frac{2\phi_{\max}^2}{(d-1)^2}\,
e^{i\pi r/d}\,\frac{\cos(\pi r/d)}{\sin^2(\pi r/d)},
\quad r = 1,\dots,d-1.
\end{aligned}
\end{equation}
This decomposition is irreducible in the sense that $\beta_r \neq 0$ for all $r \in \{1, \dots,
d-1\}$, so $\phi_x^2$ has full support on the non-trivial generalized Pauli operators
$\{Z_d^r\}_{r=1}^{d-1}$.
\end{lemma}

\begin{proof}
The set $\{Z_d^r\}_{r=0}^{d-1}$ forms a complete orthogonal basis for the $d$-dimensional vector
space of diagonal operators on $\mathbb{C}^d$ \cite{Watrous_2018}. Consequently, any diagonal
operator admits a unique expansion
\begin{equation}
\phi_x^2 = \sum_{r=0}^{d-1} \beta_r Z_d^r,
\end{equation}
with eigenvalues related by
\begin{equation}
\lambda_n^2 = \sum_{r=0}^{d-1} \beta_r \omega^{rn},
\end{equation}
i.e., the eigenvalue sequence $(\lambda_n^2)_{n=0}^{d-1}$ is the inverse discrete Fourier transform
of $(\beta_r)_{r=0}^{d-1}$. Therefore, the coefficients $\beta_r$ are given by the discrete Fourier
transform
\begin{equation}
\beta_r = \frac{1}{d} \sum_{n=0}^{d-1} \lambda_n^2 \omega^{-rn}.
\end{equation}
For the uniformly spaced grid
\(
\lambda_n = -\phi_{\max} + \frac{2\phi_{\max}}{d-1}n
\),
a direct evaluation of this sum yields
\begin{equation}
\beta_0 = \phi_{\max}^2\,\frac{d+1}{3(d-1)},
\end{equation}
and, for $r \neq 0$,
\begin{equation}\label{eq:beta_r}
\beta_r = \frac{2\phi_{\max}^2}{(d-1)^2}\,
e^{i\pi r/d}\,\frac{\cos(\pi r/d)}{\sin^2(\pi r/d)}.
\end{equation}
Since $\lambda_n^2 \in \mathbb{R}$ for all $n$, the Fourier coefficients obey
$\beta_{d-r} = \beta_r^*$, which ensures Hermiticity of $\phi_x^2$.

Finally, for $d$ odd and $r \in \{1,\dots,d-1\}$ we have $\sin(\pi r/d) \neq 0$ and
$\cos(\pi r/d) \neq 0$, so the real prefactor
is non-zero. Hence $\beta_r \neq 0$ for all $r \neq 0$, implying that $\phi_x^2$ has full support on
$\{Z_d^r\}_{r=1}^{d-1}$ and cannot be represented with fewer non-trivial generalized Pauli terms.
\end{proof}

Note that both the ``full support'' and the associated minimality 
statements are understood relative to the generalized-Pauli 
basis expansion used here, rather than claiming optimality over all possible LCU decompositions
or over alternative operator bases, which could in principle yield different representations.

\begin{definition}\label{prop:qudit-lcu-phi2}
The LCU oracles for $\phi_x^2$ are defined as
\begin{equation}
    \text{PREP}_{\phi_x^2}: |0\rangle \longmapsto \sum_{r=1}^{d-1}\sqrt{\frac{|\beta_r|}{\Lambda}}\;|r\rangle,
\end{equation}
with normalization constant \(\Lambda = \sum_{r=1}^{d-1} |\beta_r|\), and
\begin{equation}\label{eq:qudit_select}
    \text{SELECT}_{\phi_x^2} = |0\rangle\langle 0|\otimes I+\ \Big(\sum_{r=1}^{d-1}e^{i\theta_r} |r\rangle\!\langle r|\otimes Z_d^{r}\Big).
\end{equation}
\end{definition}

\begin{proof}
Note that the coefficients $\beta_r$ are in general complex. Following the standard
LCU convention (see, e.g., Refs.~\cite{Chakraborty_2024}),
we absorb the phase of each coefficient into the corresponding unitary and work
with non-negative real weights.
Therefore, we define
\begin{equation}
\tilde U_r = e^{i\theta_r}\,Z_d^{r},\qquad \theta_r=\arg\beta_r.
\end{equation}
Note that the operator can be decomposed as $\phi_x^{2} = \beta_0\mathbb{I} + \phi_x^{2\prime}$.
Thus it is sufficient to implement $e^{-it\phi_x^{2\prime}}$, since $(e^{-it\phi_x^{2\prime}})$ is
equal to $(e^{-it\phi_x^{2}})$ up to a global
phase~\cite{hardy2024optimizedquantumsimulationalgorithms}. With some abuse of notation, we drop the
identity term but continue to label the shifted operator by $(\phi_x^{2})$, i.e.,
\begin{equation}
\phi_x^{2} = \sum_{r=1}^{d-1} |\beta_r|\tilde U_r.
\end{equation}
Thus we take an index register over $r=0,\dots,d-1$, i.e., a single $d$-level qudit
(or $\lceil\log_2(d)\rceil$ qubits).
Using the standard LCU convention and defining
\(\Lambda = \sum_{r=1}^{d-1} |\beta_r|\),
we obtain
\begin{equation}
\text{PREP}_{\phi_x^2}: |0\rangle \longmapsto \sum_{r=1}^{d-1}\sqrt{\frac{|\beta_r|}{\Lambda}}\;|r\rangle
\end{equation}
and
\begin{align}
    \text{SELECT}_{\phi_x^2} &= |0\rangle\langle 0|\otimes I+\sum_{r=1}^{d-1} |r\rangle\!\langle r| \otimes \tilde U_r \nonumber \\
    &=  |0\rangle\langle 0|\otimes I+\ \Big(\sum_{r=1}^{d-1}e^{i\theta_r} |r\rangle\!\langle r|\otimes Z_d^{r}\Big).
\end{align} 
where the diagonal term acts on the index register.
\end{proof}

\subsubsection{The SELECT Oracle}\label{sec:regime2_qudit_SELECT}

With the index and system registers encoded as $d$-level qudits, the entangling part of SELECT \eqref{eq:qudit_select},
\begin{equation}
    \sum_{r=0}^{d-1} |r\rangle\langle r| \otimes Z_d^r,
\end{equation}
is the (Clifford) qudit generalized controlled-$Z$
\cite{Karacsony:2023gqc, hostens_dehaene_de_moor_qudit_stabilizer_2005}.
The only non-Clifford cost in SELECT therefore arises
from the diagonal phases that accompany the LCU coefficients, i.e.,
\begin{equation}
    D = |0\rangle\langle 0| + \sum_{r=1}^{d-1} e^{i\theta_r} |r\rangle\langle r|.
\end{equation}
For the SELECT cost analysis below, we adopt the idealized hybrid model in which code-switching between the $d$-level index qudit and an $n_b$-qubit index register is treated as free. The overhead of this conversion is discussed in Sec.~\ref{sec:regime2_codeswitch}.

\begin{corollary}\label{prop:D-diagonal-cost}
The diagonal phase operator 
\begin{equation}\label{eq:sel_diag}
    D = \sum_{r=0}^{d-1} e^{i\theta_r} |r\rangle\langle r|.
\end{equation}
can be implemented using $4n_b$ T gates, $n_b$ $R_z(\theta)$ rotations, and $n_b+1$ ancillas on an
$n_b = \lceil \log_2 d \rceil$ qubit index register.
\end{corollary}

\begin{proof}
In Appendix~\ref{app:select-cost}, we show that the Fourier coefficients of $\phi_x^2$ admit the decomposition
$\beta_r = c_r e^{i\pi r / d}$ with $c_r \in \mathbb{R}$, so that 
\begin{equation}
    D = \sum_{r=0}^{d-1} \mathrm{sgn}(c_r)\,e^{i\pi r/d}\,|r\rangle\!\langle r|
      = D_{\mathrm{sign}} D_{\mathrm{clock}}.
\end{equation}
Here $D_{\mathrm{clock}} = \exp(i \tfrac{\pi}{d} N)$, where $N$ is the number operator on the index register, and
$D_{\mathrm{sign}}$ applies a $\pm 1$ phase depending on $r$.
Using a binary expansion of $r$ and a standard reversible comparator circuit, we show that
$D_{\mathrm{clock}}$ can be implemented with $n_b$ $R_z(\theta)$ rotations
(Lemma~\ref{prop:Dclock-implementation}),
and $D_{\mathrm{sign}}$ can be implemented with $4n_b$ T gates
and $n_b+1$ ancillas (Lemma~\ref{prop:Dsign-implementation}).
\end{proof}

\subsubsection{The PREP Oracle}\label{sec:regime2_qudit_PREP}

\begin{lemma}\label{prop:prep-phi2-qudit-cost}
The PREP oracle for the qudit LCU decomposition,
\begin{equation}\label{prep_def}
    \text{PREP}_{\phi_x^2}:\; |0\rangle \longmapsto
    \sum_{r=1}^{d-1} \sqrt{\frac{|\beta_r|}{\Lambda}}\; |r\rangle,
\end{equation}
requires at most $2^{n_b}-1$ single-qubit $R_z$ rotations, up to Cliffords, when the index is encoded
in an $n_b = \lceil \log_2 d \rceil$ qubit register.
\end{lemma}

\begin{proof}

By Lemma~\ref{prop:qudit-pauli-expansion}, ${\phi}_x^2$ has full
support on the non-trivial generalized Pauli operators $\{Z_d^r\}_{r=1}^{d-1}$.
Thus the coefficients $|\beta_r|/\Lambda$ define a generic $(d-1)$-component
state on the subspace spanned by $\{|r\rangle\}_{r=1}^{d-1}$. Viewing the index
register as an $n_b = \lceil \log_2 d \rceil$-qubit system, we thus upper-bound 
$\text{PREP}_{\phi_x^2}$ as a generic real state-preparation task on that subspace.

A standard state-preparation construction uses
uniformly controlled rotations \cite{10.5555/2011670.2011675}. In the circuit decomposition for state preparation,
the ``phase equalization'' step is performed by a cascade of uniformly controlled $z$-rotations
$\Xi_z$, and the ``amplitude shaping'' step is performed by a cascade of uniformly controlled
$y$-rotations $\Xi_y$. This two-stage structure is given explicitly as a product over qubit
levels in Eq.~(7) of \cite{10.5555/2011670.2011675}.
Because the coefficients $\frac{|\beta_r|}{\Lambda}$ are real (hence no relative phases are required),
the phase-equalization stage is trivial: $\Xi_z = I$. Thus it suffices to implement
$\Xi_y^\dagger$, i.e., a cascade of uniformly controlled $R_y$ rotations.

Fix a target qubit $j\in\{1,\dots,n_b\}$. The corresponding gate in the cascade is a
$(j-1)$-fold uniformly controlled $y$-rotation $F^{\,j-1}_j(y,\cdot)$ (in the notation of
\cite{10.5555/2011670.2011675}). Note that a $k$-fold uniformly controlled
one-parameter rotation decomposes into exactly $2^k$ one-qubit elementary rotations (plus $O(k)$
CNOTs) \cite{bergholm_ucr}.
Therefore, the total number of one-qubit rotations appearing in the amplitude cascade is
\begin{equation}
\sum_{j=1}^{n_b} 2^{j-1} \;=\; 2^{n_b}-1.
\end{equation}
Finally, each single-qubit $R_y(\theta)$ can be implemented as one $R_z(\theta)$
conjugated by Clifford gates, so the non-Clifford cost is $2^{n_b}-1$ single-qubit $R_z$ rotations.
\end{proof}

\subsection{Asymptotic Scaling Comparison}\label{sec:regime2_asymptotic}
We first analyze the hybrid code-switching model in a common qubit $R_z$/$T$-count metric, which gives a clean asymptotic comparison; the fully qudit fixed-encoding analysis is deferred to Sec.~\ref{sec:regime2_nocs_crossover}.
This choice does not affect the asymptotic conclusion: the fixed-encoding block encoding still requires $\Theta(d)$ embedded two-level $SU(2)$ rotations per call (Appendix~\ref{app:reg2-crossover}), matching the $\Theta(d)$ rotation count in the code-switching construction.

Under the code-switching assumption, our non-Clifford metric is explicitly qubit $T$ gates. For a
single application of PREP and SELECT,
via Corollary~\ref{prop:D-diagonal-cost} and Lemma~\ref{prop:prep-phi2-qudit-cost}, the qudit implementation has
$
(2^{n_b}-1)+n_b
$
single-qubit $R_z$ gates, and
\begin{equation}
N^{(\mathrm{qd})}_{\mathrm{nc}}(d)=4n_b
\end{equation}
$T$ gates. By contrast, for the qubit implementation in Sec.~\ref{sec:regime2_qubit}, the full Regime~2 per-call block-encoding cost is
\begin{equation}
N^{\mathrm{LCU}}_{\mathrm{qubit}}(d,\varepsilon)
=
32\,b_r(\varepsilon)+24n_b-116,
\end{equation}
where
\begin{equation}
b_r(\varepsilon)=\left\lceil \frac12 \log_2\!\left(\frac{9\pi^2}{2\varepsilon}\right)\right\rceil.
\end{equation}
Thus, in the qubit Regime~2 baseline there is no synthesized-$R_z$ contribution.

A single block-encoding call corresponds to $W=(\mathrm{PREP}^{\dagger}\otimes I)\,\mathrm{SELECT}\,(\mathrm{PREP}\otimes I)$, so the total number of synthesized $R_z$ gates per qudit call is
\begin{equation}
L^{(2)}_{R_z,\mathrm{qd}}(d)=2(2^{n_b}-1)+n_b.
\end{equation}
Allocating the full target precision $\varepsilon$ uniformly across these synthesized qudit-side $R_z$ gates gives
\begin{equation}
\delta^{(\mathrm{qd})}(d,\varepsilon)=\frac{\varepsilon}{L^{(2)}_{R_z,\mathrm{qd}}(d)}.
\end{equation}
Using the single-qubit synthesis model
\begin{equation}\label{eq:qubit_synthesis}
C^{(\mathrm{qb})}_{R_z}(\delta)=0.57\log_2\!\left(\frac{1}{\delta}\right)+8.83,
\end{equation}
the total Regime~2 non-Clifford cost models become
\begin{equation}\label{eq:reg2_costs_qubit}
N^{\mathrm{LCU}}_{\mathrm{qubit}}(d,\varepsilon)
=
32\,b_r(\varepsilon)+24n_b-116,
\end{equation}
and
\begin{equation}\label{eq:reg2_costs_qudit}
\begin{split}
N^{\mathrm{LCU}}_{\mathrm{qudit}}(d,\varepsilon)
=
L^{(2)}_{R_z,\mathrm{qd}}(d)\\
\left[
0.57\log_2\!\left(\frac{1}{\delta^{(\mathrm{qd})}(d,\varepsilon)}\right)+8.83
\right]
+
N^{(\mathrm{qd})}_{\mathrm{nc}}(d).
\end{split}
\end{equation}

\begin{theorem}\label{prop:regime2_asymptotics}
For fixed $\varepsilon$, the cost models above imply that the qudit construction uses asymptotically more $T$ gates, i.e.
\begin{equation}
\frac{N^{\mathrm{LCU}}_{\mathrm{qubit}}(d,\varepsilon)}{N^{\mathrm{LCU}}_{\mathrm{qudit}}(d,\varepsilon)}
\longrightarrow 0
\qquad (d\to\infty).
\end{equation}
\end{theorem}

\begin{proof}
We have $N^{\mathrm{LCU}}_{\mathrm{qubit}}(d,\varepsilon)=\Theta(n_b)$ for fixed $\varepsilon$, since $b_r(\varepsilon)$ is independent of $d$ and the qubit per-call cost is $32\,b_r(\varepsilon)+24n_b-116$. By contrast, $L^{(2)}_{R_z,\mathrm{qd}}(d)=\Theta(2^{n_b})$, and hence
\begin{equation}
N^{\mathrm{LCU}}_{\mathrm{qudit}}(d,\varepsilon)
=
\Theta\!\bigl(2^{n_b} n_b\bigr)
\end{equation}
for fixed $\varepsilon$, because
\(
\log_2(1/\delta^{(\mathrm{qd})})=\log_2(L^{(2)}_{R_z,\mathrm{qd}}/\varepsilon)=\Theta(n_b).
\)
Therefore
\(
N^{\mathrm{LCU}}_{\mathrm{qubit}}(d,\varepsilon) / N^{\mathrm{LCU}}_{\mathrm{qudit}}(d,\varepsilon)\to 0,
\)
as claimed.
\end{proof}
Note that in the $d\to\infty$ limit at fixed $\phi_{\max}$ and fixed $t$, $\varepsilon_{\mathrm{sim}}$, both LCU normalizations $\alpha_{\mathrm{qb}}(d)$ and $\alpha_{\mathrm{qd}}(d)$ remain $O(1)$ in $d$ (the qudit normalization converges to a finite constant; the qubit normalization is bounded by $4\phi_{\max}^2$ but does not approach a single limit because of the dependence on $n_b=\lceil\log_2 d\rceil$). Hence the query count $Q=O(\alpha t+\log(1/\varepsilon_{\mathrm{sim}}))$ is $O(1)$ in $d$ for both encodings. Consequently, incorporating the query complexity as detailed in Sec.~\ref{sec:regime2_queries} does not change the $d$-asymptotic separation established in Theorem~\ref{prop:regime2_asymptotics}: the asymptotic ordering is determined entirely by the per-call non-Clifford cost growth.
         
\subsection{\texorpdfstring{Finite-$d$ Crossover Analysis}{Finite-d Crossover Analysis}}\label{sec:regime2_finite_d}
\subsubsection{Required number of queries}\label{sec:regime2_queries}
Given an LCU decomposition
\begin{equation}
\hat H=\sum_{j=1}^{L} \alpha_j \hat U_j,
\qquad 
\alpha = \sum_{j=1}^{L} |\alpha_j|,
\end{equation}
qubitization-based simulation of $e^{-i H t}$ to error
$\varepsilon_{\mathrm{sim}}$ uses
\begin{equation}\label{eq:query_complexity}
Q(\alpha;t,\varepsilon_{\mathrm{sim}})
=
\mathcal{O}\!\bigl(\alpha\,t+\log(1/\varepsilon_{\mathrm{sim}})\bigr)
\end{equation}
queries to a block-encoding of $\hat H$~\cite{low_chuang_qubitization_2019}, where
$\alpha=\sum_j|\alpha_j|$ is the LCU coefficient $1$-norm. A single query consists of one
application of the block-encoding unitary
\begin{equation}
W=(\mathrm{PREP}^{\dagger}\!\otimes I)\,\mathrm{SELECT}\,
(\mathrm{PREP}\otimes I).
\end{equation}
Thus the end-to-end non-Clifford cost is obtained by multiplying the per-call
block-encoding cost by the corresponding query count.

We denote by $\alpha_{\mathrm{qb}}(d)$ and $\alpha_{\mathrm{qd}}(d)$ the
LCU normalizations entering the qubit and qudit query counts, respectively.
For the qubit projector construction of Sec.~\ref{sec:regime2_qubit}, Lemma~\ref{lem:qubit_lcu_phi2}
gives
\begin{equation}\label{eq:l1norm_qubit_phi2prime}
\alpha_{\mathrm{qb}}(d)
=
(\delta_{\phi})^2
\left(2^{n_b-1}-1\right)^2 .
\end{equation}
For the qudit generalized-Pauli expansion in Definition~\ref{prop:qudit-lcu-phi2},
we discard the identity contribution $\beta_0 I$, since it contributes only a
global phase to the time evolution. The corresponding shifted-operator
normalization is therefore
\begin{equation}\label{eq:l1norm_qudit_phi2prime}
\alpha_{\mathrm{qd}}(d)
\equiv
\Lambda
=
\sum_{r=1}^{d-1}|\beta_r|,
\end{equation}
with $\beta_r$ as in Eq.~\ref{eq:beta_r}.

Following Ref.~\cite{PhysRevX.8.041015}, we define the encoding-dependent query
counts
\begin{equation}\label{eq:Q_qb_qd}
\begin{aligned}
Q_{\mathrm{qb}}(d;t,\varepsilon_{\mathrm{sim}})
&=
Q\!\bigl(\alpha_{\mathrm{qb}}(d);t,\varepsilon_{\mathrm{sim}}\bigr),\\
Q_{\mathrm{qd}}(d;t,\varepsilon_{\mathrm{sim}})
&=
Q\!\bigl(\alpha_{\mathrm{qd}}(d);t,\varepsilon_{\mathrm{sim}}\bigr).
\end{aligned}
\end{equation}
The qubit end-to-end baseline is therefore
\begin{equation}\label{eq:Ttot_qb_def}
T_{\mathrm{tot}}^{\mathrm{qb}}(d)
=
Q_{\mathrm{qb}}(d;t,\varepsilon_{\mathrm{sim}})\,
N^{\mathrm{LCU}}_{\mathrm{qubit}}
\!\bigl(d,\varepsilon_{\mathrm{BE}}^{\mathrm{qb}}(d)\bigr),
\end{equation}
where $N^{\mathrm{LCU}}_{\mathrm{qubit}}$ is the non-Clifford cost of one qubit
block-encoding call. For any specified qudit implementation with per-call cost
$N^{\mathrm{LCU}}_{\mathrm{qudit}}$, the analogous total cost is
\begin{equation}\label{eq:Ttot_qd_def}
T_{\mathrm{tot}}^{\mathrm{qd}}(d)
=
Q_{\mathrm{qd}}(d;t,\varepsilon_{\mathrm{sim}})\,
N^{\mathrm{LCU}}_{\mathrm{qudit}}
\!\bigl(d,\varepsilon_{\mathrm{BE}}^{\mathrm{qd}}(d)\bigr).
\end{equation}
In the fixed-encoding comparison below, Eq.~\ref{eq:Ttot_qb_def} gives the qubit
baseline against which the qudit synthesis threshold is defined. In the
code-switching comparison, Eqs.~\ref{eq:Ttot_qb_def}--~\ref{eq:Ttot_qd_def}
are used directly to form the total-cost ratio.

The per-call block-encoding accuracy is coupled to the target simulation error
using a per-query error budget. By a standard telescoping bound in operator norm,
$Q$ queries each accurate to $\varepsilon_{\mathrm{BE}}$ contribute at most
$Q\,\varepsilon_{\mathrm{BE}}$ additional error. It is therefore sufficient to take
\begin{equation}\label{eq:epsBE_coupled_reg2}
\varepsilon_{\mathrm{BE}}^{\mathrm{qb}}(d)
=
\frac{\varepsilon_{\mathrm{sim}}}
{Q_{\mathrm{qb}}(d;t,\varepsilon_{\mathrm{sim}})},
\qquad
\varepsilon_{\mathrm{BE}}^{\mathrm{qd}}(d)
=
\frac{\varepsilon_{\mathrm{sim}}}
{Q_{\mathrm{qd}}(d;t,\varepsilon_{\mathrm{sim}})}.
\end{equation}
For the code-switching model, the total-cost ratio is then
\begin{equation}\label{eq:Rtot_def}
\mathcal{R}_{\mathrm{tot}}(d)
=
\frac{T_{\mathrm{tot}}^{\mathrm{qb}}(d)}
{T_{\mathrm{tot}}^{\mathrm{qd}}(d)},
\end{equation}
where $\mathcal{R}_{\mathrm{tot}}(d)>1$ indicates a lower total non-Clifford cost
for the qudit-based construction.

For the numerical crossover estimates below, we instantiate
Eq.~\ref{eq:query_complexity} using the continuous proxy
\begin{equation}\label{eq:Q_proxy}
Q(\alpha;t,\varepsilon_{\mathrm{sim}})
=
\alpha\,t+\log_2\!\bigl(1/\varepsilon_{\mathrm{sim}}\bigr),
\end{equation}
rather than an exact integer query count. Since both encodings use the same
qubitization primitive, algorithm-dependent $\mathcal{O}(1)$ prefactors affect
$Q_{\mathrm{qb}}$ and $Q_{\mathrm{qd}}$ in the same way. They therefore mainly
rescale absolute costs and, through the choice
$\varepsilon_{\mathrm{BE}}=\varepsilon_{\mathrm{sim}}/Q$, enter the per-query
synthesis budgets only logarithmically. The finite-$d$ crossover thresholds
reported below should consequently be interpreted up to small $\mathcal{O}(1)$
shifts near break-even.

The proxy in Eq.~\ref{eq:Q_proxy} separates two regimes. In the
precision-dominated regime, $\log_2(1/\varepsilon_{\mathrm{sim}})$ is the larger
contribution to the query count; in the time-dominated regime, $\alpha(d)t$ is
larger. For $\varepsilon_{\mathrm{sim}}=10^{-6}$, we have
$\log_2(1/\varepsilon_{\mathrm{sim}})\approx 20$. Over the odd dimensions
considered here, $3\le d\le 1000$, one finds
$\max_d \alpha_{\mathrm{qb}}(d)\cdot 0.1\approx 0.40$ and
$\max_d \alpha_{\mathrm{qd}}(d)\cdot 0.1\approx 0.07$, so both query counts are
precision dominated at $t=0.1$. By contrast,
$\min_d \alpha_{\mathrm{qb}}(d)\cdot 3000\approx 3\times 10^3$ and
$\min_d \alpha_{\mathrm{qd}}(d)\cdot 3000\approx 2\times 10^3$, so both query
counts are time dominated at $t=3000$.

Accordingly, we use $t=0.1$ and $t=3000$ as representative precision- and
time-dominated points, respectively. Throughout this finite-$d$ analysis we fix
$\phi_{\max}=1$ and vary $d$ by refining the uniform grid spacing
$\delta_{\phi}=2\phi_{\max}/(d-1)$. These values of $t$ are not intended as
application-specific runtimes; they are chosen to isolate the two limiting
behaviors of the query model.

\subsubsection{Fixed-Encoding Model}\label{sec:regime2_nocs_crossover}

We first study the fixed-encoding model and ask how favorable the qudit
synthesis prefactor must be for a fully qudit implementation to match or beat the qubit baseline in
total non-Clifford cost.
As in Sec.~\ref{sec:regime1_crossover}, we consider prime $d$ for the qudit construction, model the synthesis cost of an embedded two-level $SU(2)$ rotation as
\begin{equation}
C_{\mathrm{prim}}(\delta)\approx a\,\log_2(1/\delta),
\end{equation}
and derive a break-even prefactor $a_{\max}^{\mathrm{LCU}}(d,t,\varepsilon_{\mathrm{sim}})$ such
that $a<a_{\max}^{\mathrm{LCU}}(d,t,\varepsilon_{\mathrm{sim}})$ implies the qudit construction is
no more expensive end-to-end.

In the fixed-encoding construction, the non-Clifford parts of the qudit block-encoding call
are expressed in terms of embedded two-level $SU(2)$ rotations. In
Appendix~\ref{app:reg2-crossover} we show that the non-Clifford part of
$\mathrm{SELECT}$ can be synthesized, up to global phase, using at most $(d-1)$ embedded two-level
$SU(2)$ rotations (Corollary~\ref{prop:D-diagonal-cost-count}), while $\mathrm{PREP}$ and
$\mathrm{PREP}^\dagger$ each require $(d-1)$ such rotations
(Lemma~\ref{prop:prep-phi2-qudit-cost-nocs}).
Therefore, for the prime local dimensions considered in the fixed-encoding analysis, the total number
of synthesized two-level rotations per block-encoding call is
\begin{equation}
L^{(2)}_{\mathrm{qd}}(d)=3d-3.
\end{equation}
Using the per-query budget from Sec.~\ref{sec:regime2_queries} and allocating it uniformly across
these rotations, we take
\begin{equation}
\delta
=
\frac{\varepsilon_{\mathrm{BE}}^{\mathrm{qd}}(d)}{L^{(2)}_{\mathrm{qd}}(d)},
\end{equation}
so that each two-level rotation is synthesized to error at most $\delta$.
This gives the per-call qudit non-Clifford estimate
\begin{align}
N^{\mathrm{LCU}}_{\mathrm{qudit}}\!\bigl(d,\varepsilon_{\mathrm{BE}}^{\mathrm{qd}}(d)\bigr)
\approx
L^{(2)}_{\mathrm{qd}}(d)\,a\,
\log_2\!\left(\frac{L^{(2)}_{\mathrm{qd}}(d)}
{\varepsilon_{\mathrm{BE}}^{\mathrm{qd}}(d)}\right).
\label{eq:Nqudit-reg2-toy}
\end{align}

To compare against the qubit Regime~2 baseline of Eq.~\ref{eq:Ttot_qb_def}, we define the largest
synthesis prefactor $a$ for which the fixed-encoding qudit construction is no more expensive in
total cost as
\begin{equation}\label{eq:amax_def_total_reg2}
a_{\max}^{\mathrm{LCU}}(d,t,\varepsilon_{\mathrm{sim}})
=
\frac{T_{\mathrm{tot}}^{\mathrm{qb}}(d)}
{Q_{\mathrm{qd}}(d)\,L^{(2)}_{\mathrm{qd}}(d)\,
\log_2\!\Big(L^{(2)}_{\mathrm{qd}}(d)/\varepsilon_{\mathrm{BE}}^{\mathrm{qd}}(d)\Big)} ,
\end{equation}
where $T_{\mathrm{tot}}^{\mathrm{qb}}(d)$, $Q_{\mathrm{qd}}(d)$, and
$\varepsilon_{\mathrm{BE}}^{\mathrm{qd}}(d)$ are defined in
Sec.~\ref{sec:regime2_queries}. Thus, if
$a<a_{\max}^{\mathrm{LCU}}(d,t,\varepsilon_{\mathrm{sim}})$, the fixed-encoding
qudit implementation uses fewer non-Cliffords than the qubit encoding within this cost model.

As in the product-formula analysis, it is useful to compare this compiler target with the effective
logarithmic prefactor that would reproduce qubit $R_z$ synthesis at the same primitive precision.
For Regime~2 fixed-encoding, the relevant primitive precision is
$\delta=\varepsilon_{\mathrm{BE}}^{\mathrm{qd}}(d)/L^{(2)}_{\mathrm{qd}}(d)$, so we define
\begin{equation}\label{eq:aRz_LCU_def}
a_{R_z}^{\mathrm{LCU}}(d,t,\varepsilon_{\mathrm{sim}})
=
\frac{
0.57\log_2\!\Big(
L^{(2)}_{\mathrm{qd}}(d)/\varepsilon_{\mathrm{BE}}^{\mathrm{qd}}(d)
\Big)+8.83
}{
\log_2\!\Big(
L^{(2)}_{\mathrm{qd}}(d)/\varepsilon_{\mathrm{BE}}^{\mathrm{qd}}(d)
\Big)
}.
\end{equation}
The comparison
\begin{equation}
a_{\max}^{\mathrm{LCU}}(d,t,\varepsilon_{\mathrm{sim}})
>
a_{R_z}^{\mathrm{LCU}}(d,t,\varepsilon_{\mathrm{sim}})
\end{equation}
indicates that the fixed-encoding qudit construction can tolerate embedded-$SU(2)$ synthesis at least
as costly as qubit $R_z$ synthesis at the same primitive precision. Conversely, if
$a_{\max}^{\mathrm{LCU}}<a_{R_z}^{\mathrm{LCU}}$, then a genuine per-primitive qudit synthesis
advantage is required.

In the precision-dominated regime, with $\varepsilon_{\mathrm{sim}}=10^{-6}$ and $t=0.1$,
Table~\ref{tab:fixed_encoding_amax_short} reports representative values of
$a_{\max}^{\mathrm{LCU}}(d,0.1,10^{-6})$. The
condition
$a_{\max}^{\mathrm{LCU}}>a_{R_z}^{\mathrm{LCU}}$
holds only for $d=3$ and $d=5$. At $d=7$ the two quantities are already nearly equal, with
$a_{\max}^{\mathrm{LCU}}/a_{R_z}^{\mathrm{LCU}}\approx 0.97$, and for larger prime dimensions the
fixed-encoding qudit construction requires an increasingly favorable embedded-$SU(2)$ synthesis
prefactor.

\begin{table}[t]
\centering
\begin{tabular}{c c}
\hline\hline
\(d\)
&
\(a^{\mathrm{LCU}}_{\max}(d,0.1,10^{-6})\)
\\
\hline
3  & 2.56 \\
5  & 1.32 \\
7  & 0.85 \\
11 & 0.53 \\
13 & 0.44 \\
17 & 0.34 \\
19 & 0.30 \\
\hline\hline
\end{tabular}
\caption{
End-to-end break-even synthesis prefactor
\(a^{\mathrm{LCU}}_{\max}(d,t,\varepsilon_{\mathrm{sim}})\)
for the fixed-encoding Regime~2 qudit LCU under the coupled per-query error budget, evaluated at
\(\varepsilon_{\mathrm{sim}}=10^{-6}\) and \(t=0.1\). Only prime local dimensions \(d\) up to
\(19\) are included in the fixed-encoding analysis.}
\label{tab:fixed_encoding_amax_short}
\end{table}

In the time-dominated regime, with $\varepsilon_{\mathrm{sim}}=10^{-6}$ and $t=3000$, the relative
query normalizations make the fixed-encoding comparison more favorable, as shown in Fig.~\ref{fig:regime2_fixed_amax_t3000}.
With the same-precision reference prefactor
$a_{R_z}^{\mathrm{LCU}}$, the favorable prime dimensions are those for which
$a_{\max}^{\mathrm{LCU}}>a_{R_z}^{\mathrm{LCU}}$:
\begin{equation}
d\in\{3,5,7,11,13,17,19\}.
\end{equation}
The largest value in this regime occurs at $d=5$, where
$a_{\max}^{\mathrm{LCU}}(5,3000,10^{-6})=4.794611$ and
$a_{R_z}^{\mathrm{LCU}}(5,3000,10^{-6})=0.825901$.
The last favorable prime dimension is $d=19$, where
$a_{\max}^{\mathrm{LCU}}(19,3000,10^{-6})=1.339724$ and
$a_{R_z}^{\mathrm{LCU}}(19,3000,10^{-6})=0.810783$.
For $d\ge 23$ among the prime dimensions considered, the break-even prefactor lies below the
same-precision qubit-\(R_z\) reference prefactor.

\begin{figure}[t]
  \centering
  \includegraphics[width=\linewidth]{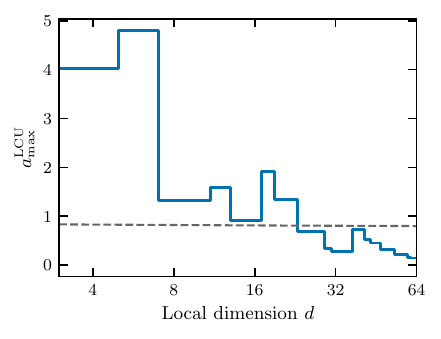}
  \caption{Break-even synthesis prefactor
  \(a_{\max}^{\mathrm{LCU}}(d,3000,10^{-6})\)
  for the fixed-encoding Regime~2 implementation under the coupled per-query error budget. The
  dashed curve marks the same-precision qubit-\(R_z\) effective prefactor
  \(a_{R_z}^{\mathrm{LCU}}(d,3000,10^{-6})\) from Eq.~\ref{eq:aRz_LCU_def}. Values of
  \(a_{\max}^{\mathrm{LCU}}\) above this reference prefactor indicate that the fully qudit implementation can
  match or outperform the qubit baseline in total non-Clifford cost without requiring an embedded
  two-level qudit synthesis prefactor smaller than the corresponding qubit-\(R_z\) effective
  prefactor. Only prime local dimensions \(d\) are included in the fixed-encoding analysis.}
  \label{fig:regime2_fixed_amax_t3000}
\end{figure}

\subsubsection{Idealized Code-Switching Model}\label{sec:regime2_codeswitch}
We turn to the code-switching model.
Using the end-to-end proxies $T_{\mathrm{tot}}^{\mathrm{qb/qd}}(d)$ and the ratio
$\mathcal{R}_{\mathrm{tot}}(d)$ defined in Sec.~\ref{sec:regime2_queries}, we consider the
idealized hybrid setting in which the $n_b$-qubit index register can be code-switched to a single
$d$-level qudit (and back) to implement the controlled-$Z_d$ interaction.

\begin{figure*}[t]
\centering
\includegraphics[width=\textwidth]{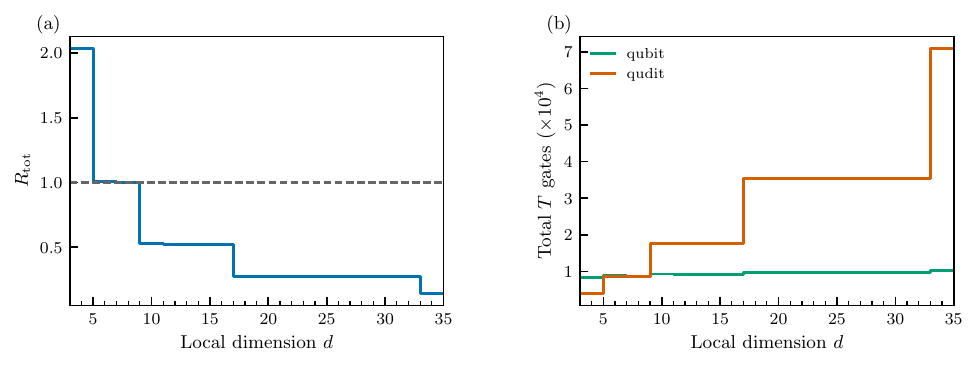}
\caption{Non-Clifford resource comparison for the Regime~2 LCU/block-encoding implementation under the code-switching model, for $\varepsilon_{\mathrm{sim}}=10^{-6}$ and $t=0.1$. Panel (a) shows the ratio
$\mathcal{R}_{\mathrm{tot}}(d)=T_{\mathrm{tot}}^{\mathrm{qb}}(d)/T_{\mathrm{tot}}^{\mathrm{qd}}(d)$
as a function of the local dimension $d$. Panel (b) shows the corresponding total $T$-count estimates
$T_{\mathrm{tot}}^{\mathrm{qb}}(d)$ and $T_{\mathrm{tot}}^{\mathrm{qd}}(d)$.}
\label{fig:regime2_t0p1}
\end{figure*}

Figure~\ref{fig:regime2_t0p1} shows the resulting total $T$-counts, $T_{\mathrm{tot}}^{\mathrm{qb}}(d)$ and $T_{\mathrm{tot}}^{\mathrm{qd}}(d)$, together with the corresponding ratio $\mathcal{R}_{\mathrm{tot}}(d)$ at $\varepsilon_{\mathrm{sim}}=10^{-6}$ and $t=0.1$. For these parameters, the qudit construction is cheaper only for $d=3$ and $d=5$, while the qubit construction is cheaper for $d\ge 7$. The crossover occurs between $d=5$ and $d=7$: $\mathcal{R}_{\mathrm{tot}}(5)=1.006205$ and $\mathcal{R}_{\mathrm{tot}}(7)=0.999963$. The piecewise-constant structure in $d$ is a consequence of the dependence on $n_b=\lceil \log_2 d\rceil$. The largest relative qudit advantage occurs at $d=3$, where $\mathcal{R}_{\mathrm{tot}}(3)=2.033787$, and the largest absolute savings also occur at $d=3$, where the qudit construction saves $T_{\mathrm{tot}}^{\mathrm{qb}}(3)-T_{\mathrm{tot}}^{\mathrm{qd}}(3)\approx 4.20\times 10^3$ $T$ gates.

\begin{figure*}[t]
\centering
\includegraphics[width=\textwidth]{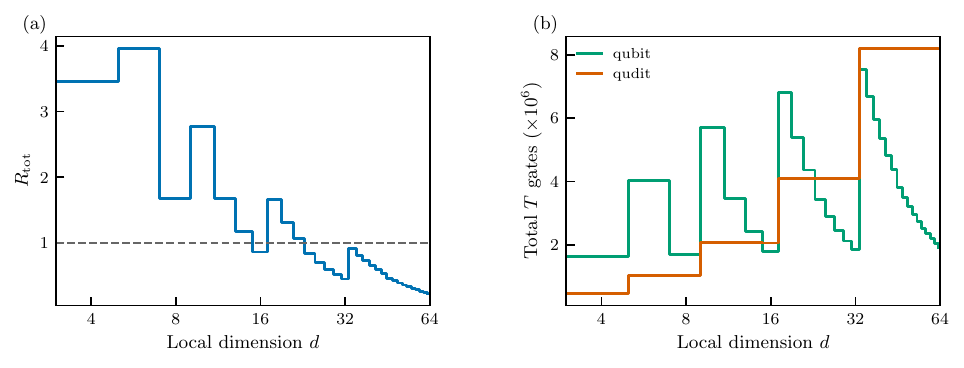}
\caption{Non-Clifford resource comparison for the Regime~2 LCU/block-encoding implementation under the code-switching model, for $\varepsilon_{\mathrm{sim}}=10^{-6}$ and $t=3000$. Panel (a) shows the ratio
$\mathcal{R}_{\mathrm{tot}}(d)=T_{\mathrm{tot}}^{\mathrm{qb}}(d)/T_{\mathrm{tot}}^{\mathrm{qd}}(d)$
as a function of the local dimension $d$. Panel (b) shows the corresponding total $T$-count estimates
$T_{\mathrm{tot}}^{\mathrm{qb}}(d)$ and $T_{\mathrm{tot}}^{\mathrm{qd}}(d)$.
The piecewise-constant structure and repeated crossover windows arise from the binary embedding.}
\label{fig:regime2_t3000}
\end{figure*}

Figure~\ref{fig:regime2_t3000} shows the resulting total $T$-counts, $T_{\mathrm{tot}}^{\mathrm{qb}}(d)$ and $T_{\mathrm{tot}}^{\mathrm{qd}}(d)$, together with the corresponding ratio $\mathcal{R}_{\mathrm{tot}}(d)$ at $\varepsilon_{\mathrm{sim}}=10^{-6}$ and $t=3000$ (time-dominated regime). For these parameters, the qudit construction is cheaper for
$d\in [3,13]\cup[17,21]$,
while the qubit construction is cheaper outside these windows. The final crossover occurs between $d=21$ and $d=23$: $\mathcal{R}_{\mathrm{tot}}(21)=1.062653$, whereas $\mathcal{R}_{\mathrm{tot}}(23)=0.835319$. As in the precision-dominated comparison, the piecewise structure in $d$ originates from the dependence on $n_b=\lceil \log_2 d\rceil$. The largest relative qudit advantage occurs at $d=5$, where $\mathcal{R}_{\mathrm{tot}}(5)=3.959978$, while the largest absolute savings occur at $d=9$, where the qudit construction saves $T_{\mathrm{tot}}^{\mathrm{qb}}(9)-T_{\mathrm{tot}}^{\mathrm{qd}}(9)\approx 3.65\times 10^6$ $T$ gates.

We can also interpret the qudit advantage as a code-switching budget.
Let
\begin{equation}\label{eq:codeswitch_budget}
\Delta_{\mathrm{tot}}(d)
=
T_{\mathrm{tot}}^{\mathrm{qb}}(d)-T_{\mathrm{tot}}^{\mathrm{qd}}(d),
\end{equation}
and suppose the implementation requires $k$ directional code-switches per query.
In the hybrid construction considered here, the $n_b$-qubit index register is used for PREP (and the diagonal
$R_z$-synthesis term in SELECT); it is then converted to a single $d$-level qudit to implement the controlled-$Z_d$
interaction, and finally converted back to qubits, so $k=2$.
If each switch incurs non-Clifford overhead $T_{\mathrm{cs}}$, the added overhead is
$Q_{\mathrm{qd}}(d)\,k\,T_{\mathrm{cs}}$, where $Q_{\mathrm{qd}}(d)$ is the qudit query count from
Sec.~\ref{sec:regime2_queries}.
Thus, qudits remain advantageous provided
\begin{equation}\label{eq:codeswitch_ineq}
T_{\mathrm{cs}}
<
\frac{\Delta_{\mathrm{tot}}(d)}{Q_{\mathrm{qd}}(d)\,k}.
\end{equation}

For $\varepsilon_{\mathrm{sim}}=10^{-6}$ and $t=0.1$ with $k=2$, this yields representative budgets $T_{\mathrm{cs}}\approx 1.05\times 10^{2}$ (for $d=3$) and $T_{\mathrm{cs}}\approx 1.35\times 10^{0}$ (for $d=5$); the budget is already negative by $d=7$, reflecting the fact that the total-cost crossover has already occurred.

For $\varepsilon_{\mathrm{sim}}=10^{-6}$ and $t=3000$, this yields representative budgets $T_{\mathrm{cs}}\approx 2.87\times 10^{2}$ (for $d=3$), $T_{\mathrm{cs}}\approx 7.42\times 10^{2}$ (for $d=5$), $T_{\mathrm{cs}}\approx 8.97\times 10^{2}$ (for $d=9$), and $T_{\mathrm{cs}}\approx 6.65\times 10^{2}$ (for $d=17$), while near the final crossover the budget drops to $T_{\mathrm{cs}}\approx 6.34\times 10^{1}$ at $d=21$ and becomes negative by $d=23$.
We discuss the implications of such a comparison in Sec.~\ref{sec:discussion}.

\section{Discussion}
\label{sec:discussion}
We compare fault-tolerant qudit and qubit encodings for implementing the onsite diagonal
operator $e^{-it\phi_x^2}$ arising from a $d$-level field-amplitude truncation at matched
accuracy. The central question is not whether qudits win asymptotically for this operator
class, but whether finite-$d$ windows exist in which they can reduce non-Clifford cost.
Our results show that such windows do exist, but their size and relevance depend strongly on
the simulation regime, the achievable synthesis cost of embedded two-level qudit primitives,
and, in hybrid constructions, the overhead of code-switching.

Table~\ref{tab:asymptotic_summary} summarizes the leading-order non-Clifford scalings
 established below for both regimes and encodings, within the constructive cost models analyzed here.

\begin{table*}[t]
\centering
\setlength{\tabcolsep}{6pt}
\renewcommand{\arraystretch}{1.25}
{
\begin{tabular}{@{}lcc@{}}
\hline\hline
Cost & Qubit & Qudit \\
\hline
PF, rotations per step                              & $\Theta(n_b^{2})$                                          & $\Theta(d)$ \\
PF, non-Clifford per step, total step error $\nu$   & $\Theta\!\big(n_b^{2}\log\tfrac{n_b^{2}}{\nu}\big)$        & $\Theta\!\big(d\log\tfrac{d}{\nu}\big)$ \\
\hline
LCU, rotations per block-encoding call              & $0$                                                       & $\Theta(d)$ \\
LCU, non-Clifford per call, total call error $\varepsilon$ & $\Theta\!\big(n_b+\log\tfrac{1}{\varepsilon}\big)$   & $\Theta\!\big(d\log\tfrac{d}{\varepsilon}\big)$ \\
\hline\hline
\end{tabular}}
\caption{Leading-order asymptotic  scalings for a single onsite quadratic diagonal evolution $e^{-it\phi_x^2}$,
 comparing the $n_b=\lceil\log_2 d\rceil$-qubit encoding with the native $d$-level qudit encoding, 
 in the product-formula (PF, per Trotter step) and LCU (LCU, per block-encoding call) regimes. 
 ``Rotations'' counts synthesized one-parameter rotations ($R_z$ for qubits, embedded two-level $SU(2)$ for qudits);
 the qubit LCU baseline uses Toffoli/$T$ gates rather than synthesized rotations.
 Here $\nu$ and $\varepsilon$ are errors for one Trotter step and for one
block-encoding call, respectively.}
\label{tab:asymptotic_summary}
\end{table*}

In Regime~1, the binary qubit encoding has an asymptotic advantage for structural reasons:
the binary expansion of the onsite quadratic phase yields only one- and two-body $Z$ and
$ZZ$ terms, and hence only $\Theta((\log d)^2)$ synthesized $R_z$ rotations per Trotter
step. By contrast, the single-qudit implementation analyzed in Sec.~\ref{sec:regime1_qudit}
uses exactly $d-1$ nontrivial embedded two-level $R_Z^{(k,k+1)}(\theta_k)$ rotations for the
symmetric truncation. Under the common logarithmic-in-precision synthesis model, this means
that for the qudit route to asymptotically beat the qubit route, the per-rotation synthesis
prefactor for the embedded $SU(2)$ rotations would need to be smaller than that of a qubit
$R_z$ by a factor $\Theta((\log d)^2/d)$, i.e., an exponentially strong improvement in
$n_b=\lceil\log_2 d\rceil$. Since known fault-tolerant synthesis procedures already achieve
$O(\log(1/\delta))$ scaling for common discrete gate
sets~\cite{ross_selinger_optimal_rz_2014,Bocharov:2016yvg,gustafson2025synthesissinglequtritcircuits,10.5555/2685188.2685198},
we do not expect such an exponential per-rotation advantage to arise. Likewise, phase
kickback does not evade the basic issue: after coherently computing $n^2$, one must still
synthesize a diagonal phase function of that value.

The finite-$d$ Regime~1 crossover analysis points in the same direction. Modeling each
embedded qudit $SU(2)$ primitive as costing $\sim a\log_2(1/\delta)$ non-Cliffords and
comparing against the qubit $R_z$ synthesis model $0.57\log_2(1/\delta)+8.83$
yields the break-even threshold $a_{\max}^{\mathrm{PF}}(d,\varepsilon)$ and the qubit-$R_z$
effective prefactor $a_{R_z}^{\mathrm{PF}}(d,\varepsilon)$ (Sec.~\ref{sec:regime1_crossover}).
For representative $\varepsilon=10^{-6}$, we find
$a_{\max}^{\mathrm{PF}}(3,10^{-6})\approx 1.51$ and
$a_{\max}^{\mathrm{PF}}(5,10^{-6})\approx 1.48$,
while $a_{\max}^{\mathrm{PF}}(7,10^{-6})\approx 0.96$ already falls below
$a_{R_z}^{\mathrm{PF}}$. The condition $a_{\max}^{\mathrm{PF}}>a_{R_z}^{\mathrm{PF}}$
holds only for $d=3$ and $d=5$ among prime dimensions at this precision.
We therefore view Regime~1 as allowing, at most, a narrow low-$d$ opportunity for qudits,
rather than a robust route to scalable savings.

In Regime~2, the comparison is between the native qudit LCU construction and the standard signed-binary qubit projector-LCU baseline of Refs.~\cite{Su_2021_PRXQ,Spagnoli_2025}, both targeting the same physical $d=2M+1$ field-amplitude truncation. The qudit construction makes SELECT particularly cheap: when both index and
system are encoded as qudits, the entangling part $\sum_r |r\rangle\langle r|\otimes Z_d^r$
is a generalized controlled-$Z$ and hence Clifford, leaving only a simple diagonal on the
index register as the non-Clifford component. However, the qubit baseline still has
an asymptotic advantage overall, because in the generalized-Pauli expansion on a $d$-level
system, $\phi_x^2$ has irreducible LCU length $d-1$, so the corresponding PREP oracle
amounts to generic $(d-1)$-dimensional state preparation requiring $d-1$ embedded two-level
rotations (Lemma~\ref{prop:prep-phi2-qudit-cost-nocs}). By contrast, the qubit baseline's
per-call cost grows only as $\Theta(n_b)$ (Lemma~\ref{lem:qubit_lcu_phi2}). Thus, while
SELECT is especially favorable in the qudit encoding, the full oracle cost is not
asymptotically improved (Theorem~\ref{prop:regime2_asymptotics}). As elsewhere in Regime~2, this is a construction-specific statement about the standard signed-binary qubit projector-LCU baseline; alternative qubit LCU constructions could change the asymptotics.

We therefore analyze finite-$d$ behavior in an end-to-end sense by weighting per-block-encoding
costs by the qubitization query complexity (Sec.~\ref{sec:regime2_nocs_crossover}). In the
fixed-encoding setting, the precision-dominated regime ($\varepsilon_{\mathrm{sim}}=10^{-6}$,
$t=0.1$) is relatively restrictive: as shown in Table~\ref{tab:fixed_encoding_amax_short}, the
condition $a_{\max}^{\mathrm{LCU}}>a_{R_z}^{\mathrm{LCU}}$ holds only for $d=3$
($a_{\max}^{\mathrm{LCU}}=2.56$) and $d=5$ ($a_{\max}^{\mathrm{LCU}}=1.32$), while
$a_{\max}^{\mathrm{LCU}}$ falls below the qubit-$R_z$ reference prefactor for $d\ge 7$.

In the time-dominated regime ($\varepsilon_{\mathrm{sim}}=10^{-6}$, $t=3000$), the picture is
more favorable to qudits: $a_{\max}^{\mathrm{LCU}}(d,3000,10^{-6})$ exhibits multiple
threshold crossings as a function of $d$, reflecting the piecewise dependence on the binary
embedding through $n_b=\lceil\log_2 d\rceil$ (Fig.~\ref{fig:regime2_fixed_amax_t3000}).
The condition $a_{\max}^{\mathrm{LCU}}>a_{R_z}^{\mathrm{LCU}}$ holds for all prime dimensions
$d\in\{3,5,7,11,13,17,19\}$, with the largest value $a_{\max}^{\mathrm{LCU}}(5,3000,10^{-6})=4.794611$
and the last favorable prime $d=19$, where $a_{\max}^{\mathrm{LCU}}(19,3000,10^{-6})=1.339724$
compared with $a_{R_z}^{\mathrm{LCU}}(19,3000,10^{-6})=0.810783$.
For $d\ge 23$ among prime dimensions, the break-even prefactor lies below the qubit-$R_z$
reference. These values are not uniform in $d$, so any practical advantage remains
truncation-dependent.

For context, existing qutrit synthesis results report effective logarithmic prefactors
around $\sim 3.16$--$3.24$ (in $R$-count) for relevant diagonal/two-level
primitives~\cite{gustafson2025synthesissinglequtritcircuits,Bocharov:2016yvg}. Because the
fixed-encoding break-even thresholds $a_{\max}^{\mathrm{LCU}}$ are defined relative to the
qubit $T$-gate synthesis model, a direct numerical comparison with $R$-count data requires
a gate-set conversion not available in general; nevertheless, the favorable $d=3$ window in the
time-dominated regime has a break-even prefactor substantially above the qubit reference
$a_{R_z}^{\mathrm{LCU}}\approx 0.83$, suggesting that concrete qutrit synthesis data could
fall within the achievable range. Further, for odd prime dimensions,
Ref.~\cite{PrakashKalraJain2021} gives an asymptotic worst-case lower bound for approximating
arbitrary $SU(p)$ unitaries by single-qudit Clifford+$T$ circuits. This bound sets a floor
on achievable synthesis cost and provides context for whether the break-even thresholds
identified here are plausibly reachable, but a definitive comparison requires dimension-specific
synthesis results.

We then consider the code-switching setting, which we emphasize is entirely idealized:
determining the fault-tolerant overhead of converting between a logical qudit and a logical
qubit register remains an open problem. In the precision-dominated regime
($t=0.1$, $\varepsilon_{\mathrm{sim}}=10^{-6}$, Fig.~\ref{fig:regime2_t0p1}), the qudit
construction is cheaper only for $d=3$ and $d=5$, with the crossover occurring between
$d=5$ and $d=7$: $\mathcal{R}_{\mathrm{tot}}(5)=1.006205$ while
$\mathcal{R}_{\mathrm{tot}}(7)=0.999963$. The largest relative advantage occurs at $d=3$,
where $\mathcal{R}_{\mathrm{tot}}(3)=2.033787$, and the largest absolute savings also occur
at $d=3$, where the qudit construction saves approximately $4.20\times10^3$ $T$ gates.

In the time-dominated regime ($t=3000$, $\varepsilon_{\mathrm{sim}}=10^{-6}$,
Fig.~\ref{fig:regime2_t3000}), the end-to-end comparison exhibits multiple crossings tied to
the binary embedding, with the qudit construction cheaper for
$d\in[3,13]\cup[17,21]$ and the qubit encoding cheaper outside these windows.
The final crossover occurs between $d=21$ and $d=23$: $\mathcal{R}_{\mathrm{tot}}(21)=1.062653$
while $\mathcal{R}_{\mathrm{tot}}(23)=0.835319$.
The largest relative advantage occurs at $d=5$, where $\mathcal{R}_{\mathrm{tot}}(5)=3.959978$,
while the largest absolute savings occur at $d=9$, where the qudit construction saves
approximately $3.65\times10^6$ $T$ gates.

Equivalently, these savings define a loose break-even budget for code-switching. Writing
$\Delta_{\mathrm{tot}}(d)=T_{\mathrm{tot}}^{\mathrm{qb}}(d)-T_{\mathrm{tot}}^{\mathrm{qd}}(d)$
and assuming $k=2$ directional code switches per query, qudit implementations remain
advantageous provided the additional non-Clifford overhead per switch satisfies
Eq.~\ref{eq:codeswitch_ineq}. In the precision-dominated regime, representative budgets are
$T_{\mathrm{cs}}\approx 1.05\times10^2$ T gates per switch at $d=3$ and
$T_{\mathrm{cs}}\approx 1.35$ at $d=5$; the budget is already negative at $d=7$.
In the time-dominated regime, budgets reach up to $\mathcal{O}(10^3)$ T gates per switch
near the most favorable dimensions (e.g., $T_{\mathrm{cs}}\approx 8.97\times10^2$ at $d=9$),
declining toward the final crossover ($T_{\mathrm{cs}}\approx 6.34\times10^1$ at $d=21$,
negative for $d\ge23$).

However, comparing $T_{\mathrm{cs}}$ with the actual fault-tolerant cost of code-switching in
this setting is highly nontrivial. The overhead of converting between a fault-tolerant
$d$-level qudit encoding and a fault-tolerant qubit register is strongly architecture- and
code-dependent, and a concrete implementation with explicit space--time accounting remains an
open problem. Relatedly, Ref.~\cite{Moussa:2016jar} introduced qubit fusion
and qudit fission, effecting a conversion between two qubits and a single $d=4$ qudit via
stabilizer circuits and the consumption of a non-stabilizer resource state $|F\rangle$. In
that framework, a natural route to a $T$-equivalent estimate is to compare the cost of
supplying high-fidelity $|F\rangle$ states with that of supplying high-fidelity qubit magic
states. However, the absolute and relative costs of $|F\rangle$ versus $|T\rangle$ depend
sensitively on the available distillation protocols and the underlying fault-tolerant
architecture. Consequently, a full end-to-end resource estimate for code
switching---including state injection, distillation, and the associated space--time overhead
of the conversion circuitry---is beyond the scope of this work.

The small-$d$ window of potential advantage identified in Regime~1 and in the
precision-dominated Regime~2 comparison is the most accessible regime for near-term
qudit synthesis efforts. The time-dominated Regime~2 model extends this window further:
in the fixed-encoding comparison, the break-even condition $a_{\max}^{\mathrm{LCU}}>a_{R_z}^{\mathrm{LCU}}$
holds for prime dimensions up to $d=19$; in the code-switching comparison, the favorable
range extends to $d\le 21$.
Ref.~\cite{klco_savage_scalar_digitization_2018} argues that roughly
$4$--$7$ qubits per site (corresponding to $d=2^4$--$2^7$) are
required before digitization error in $\phi_x$ is reduced below typical Trotter error.
This scale is at or above the end of the favorable windows identified here: only the lower
end of this range ($d\approx 16$--$19$) overlaps with the time-dominated break-even region,
while larger truncations ($d\ge 23$) lie outside the favorable window in both models.
Furthermore, diagonal evolutions generated by quadratic onsite/link
operators of the same form arise naturally in other settings, where the physically relevant
Hilbert-space cutoffs are often explicitly small integers. For example, for gauge links in
$(1+1)$D scalar QED, a truncation $d\leq 13$ appears sufficient~\cite{PhysRevB.103.245137,Gustafson_2021}.
For pure $U(1)$ gauge theory in $(2+1)$D, Ref.~\cite{PhysRevD.107.L031503} finds
per-mille plaquette accuracy with only $d=7$ states per lattice site; and for 2D lattice
QED plaquette/gauge-link simulations, Refs.~\cite{Paulson:2020zjd,Meth:2023wzd} explicitly
work with minimal-to-improved gauge-field discretizations at $d=3$ and $d=5$, with
convergence trends indicating that these modest $d$ already capture the relevant low-energy
physics. These examples further motivate investigations of fault-tolerant qudits for
$d>3$.

Finally, we comment on how far these conclusions may extend beyond the specific operator $\phi_x^2$. For a fixed-degree polynomial diagonal term $h(N)$ in a truncated integer-valued operator $N$, the Regime~1 asymptotic lesson is expected to persist under the same constructive route: the binary qubit expansion yields a number of $Z$-type rotations that grows polynomially in $n_b=\lceil\log_2 d\rceil$, whereas the direct embedded-rotation qudit diagonal implementation of Lemma~\ref{prop:su2_irred} generally uses $O(d)$ two-level rotations. The quadratic case is therefore not uniquely responsible for the asymptotic ordering. The finite-$d$ crossover thresholds, by contrast, are operator-dependent, so the specific break-even values reported here would need to be recomputed for higher-order potentials. For generic, unstructured diagonal operators we claim no universal conclusion: the comparison can become sensitive to the choice of basis, oracle construction, and synthesis model. Conversely, some lattice-gauge-theory settings, for example $\mathbb{Z}_N$ lattice gauge theories,
 carry qudit structure that may structurally favour qudit-based constructions: recent work exploiting prime-$N$ $\mathbb{Z}_N$ lattice-gauge-theory structure within a qudit stabiliser-code framework is one example~\cite{spagnoli2026quditstabilisercodes}. We therefore regard the present quadratic analysis as a controlled benchmark for qudit-versus-qubit non-Clifford costs, rather than a universal statement about all diagonal Hamiltonians.

\section{Conclusions}
\label{sec:conclusions}

Overall, these results suggest that fault-tolerant qudits should be viewed not as providing a generic asymptotic advantage for this class of quadratic diagonal evolutions, but rather as a potentially useful finite-$d$ resource when the truncation, simulation regime, and achievable synthesis overhead are favorable. Within the constructions studied here, the LCU/block-encoding (Regime~2) comparison is the more promising setting: the fixed-encoding analysis identifies break-even conditions at low prime dimensions (favorable up to $d=19$ in the time-dominated regime), while the idealized code-switching analysis shows that meaningful constant-factor savings can arise at selected truncations up to $d\lesssim 21$, subject to architecture-dependent conversion costs. The Regime~2 comparison is, in particular, a construction-specific comparison between the native qudit LCU construction and the standard signed-binary qubit projector-LCU baseline of Refs.~\cite{Su_2021_PRXQ,Spagnoli_2025}; alternative qubit LCU constructions could change the comparison in either direction.

We therefore interpret the break-even thresholds derived here as concrete compiler targets rather than as a blanket demonstration of qudit superiority. A more definitive comparison will require explicit gate-set-dependent synthesis costs for diagonal $SU(d)$ operations, or equivalently for the embedded two-level qudit $SU(2)$ primitives used throughout, together with architecture-level resource estimates for code conversion, state injection and distillation, and associated space-time overheads. It would also be valuable to extend this analysis to other operator classes for which the structural simplifications available to binary qubit encodings for diagonal operators do not apply.

\begin{acknowledgments}
We thank Joseph Carlson, Minoo Kabirnezhad, Andy C. Y. Li, and Constantinos Andreopoulos for valuable conversations, as well as Cameron Cook for helpful comments on the draft. S.G. acknowledges an STFC PhD studentship
at the LIV.INNO Centre for Doctoral Training ``Innovation in Data
Intensive Science''.
G.P. and D.K. were partially supported for this work by the DOE/HEP QuantISED program grant ``HEP Machine Learning and Optimization Go Quantum,'' identification number 0000240323. M.M. acknowledges the support of Agencia Nacional de  Investigación e Innovación, project POS\_EXT\_2023\_1\_175687, and LIV.INNO.

This work was produced by Fermi Forward Discovery Group, LLC under Contract No. 89243024CSC000002 with the U.S. Department of Energy, Office of Science, Office of High Energy Physics. The United States Government retains and the publisher, by accepting the work for publication, acknowledges that the United States Government retains a non-exclusive, paid-up, irrevocable, world-wide license to publish or reproduce the published form of this work, or allow others to do so, for United States Government purposes. The Department of Energy will provide public access to these results of federally sponsored research in accordance with the DOE Public Access Plan (\url{http://energy.gov/downloads/doe-public-access-plan}).

Fermilab report number: FERMILAB-PUB-26-0287-ETD.
\end{acknowledgments}

%% file: appendix.tex
\section{SELECT-Oracle Cost Details}\label{app:select-cost}

\begin{lemma}\label{prop:D-decomposition}
The diagonal phase operator
\begin{equation}
    D = \sum_{r=0}^{d-1} e^{i\theta_r}\,|r\rangle\langle r|
\end{equation}
with $\theta_0 = 0$, can be written as
\begin{equation}
D = \sum_{r=0}^{d-1} \mathrm{sgn}(c_r)\,e^{i\pi r/d}\,|r\rangle\langle r|
   = D_{\mathrm{sign}} D_{\mathrm{clock}},
\end{equation}
where $c_r \in \mathbb{R}$ and
\begin{equation}
\begin{aligned}
D_{\mathrm{clock}} &= \sum_{r=0}^{d-1} e^{i\pi r/d}\,|r\rangle\langle r|, \\
D_{\mathrm{sign}} &= |0\rangle \langle 0| + \sum_{r=1}^{d-1} \mathrm{sgn}(c_r)\,|r\rangle\langle r|.
\end{aligned}
\end{equation}
\end{lemma}

\begin{proof}
The operator $\phi_x^2$ is diagonal with real eigenvalues $\lambda_n^2$, so its expansion in the
generalized Pauli basis is
\begin{equation}
    \phi_x^2 = \sum_{r=0}^{d-1} \beta_r Z_d^r,
\end{equation}
with Fourier coefficients
\begin{equation}
\beta_r = \frac{1}{d}\sum_{n=0}^{d-1} \lambda_n^2\,\omega^{-rn}, \qquad \omega = e^{2\pi i/d}.
\end{equation}
For the uniform symmetric truncation used here, the explicit evaluation of this sum
(Eq.~\ref{eq:coeff_beta_decomp}) gives, for $r \neq 0$,
\begin{equation}
    \beta_r = \frac{2\phi_{\max}^2}{(d-1)^2}\,
    e^{i\pi r/d}\,\frac{\cos(\pi r/d)}{\sin^2(\pi r/d)}
    = c_r e^{i\pi r/d},
\end{equation}
where
\begin{equation}
    c_r = \frac{2\phi_{\max}^2}{(d-1)^2}\,\frac{\cos(\pi r/d)}{\sin^2(\pi r/d)} \in \mathbb{R}.
\end{equation}
Thus each $\beta_r$ has a fixed phase $e^{i\pi r/d}$ and a real amplitude $c_r$.

Defining the phases of $D$ by $e^{i\theta_r} = \beta_r/|\beta_r|$ and using
$\mathrm{sgn}(c_r) = c_r/|c_r|$, we obtain
\begin{equation}\label{eq:beta_phases}
e^{i\theta_r} = \frac{\beta_r}{|\beta_r|}
= \frac{c_r e^{i\pi r/d}}{|c_r|}
= \mathrm{sgn}(c_r)\,e^{i\pi r/d},
\end{equation}
so
\begin{equation}
D = \sum_{r=0}^{d-1} e^{i\theta_r}|r\rangle\langle r|
   = D_{\mathrm{sign}} D_{\mathrm{clock}},
\end{equation}
as claimed.
\end{proof}

\begin{lemma}\label{prop:Dclock-implementation}
The operator $D_{\mathrm{clock}} = \exp\!\big(i \frac{\pi}{d} N\big)$ can be implemented using
$O(n_b)$ single-qubit $Z$-rotations on the index register, i.e., $O(n_b)$ non-Clifford gates.
\end{lemma}

\begin{proof}
We first write the integer $r$ in binary as
\begin{equation}
    r = \sum_{m=0}^{n_b-1} 2^m q_m, \qquad q_m \in \{0,1\},
\end{equation}
where $q_m$ is the bit stored on qubit $m$ of the index register. In terms of Pauli operators,
\begin{equation}
    q_m = \frac{1 - Z^{(m)}}{2},
\end{equation}
where $Z^{(m)}$ is the Pauli-$Z$ on qubit $m$. Substituting this into the definition of the number operator, we obtain
\begin{equation}
    N = \sum_{m=0}^{n_b-1} 2^m q_m = \frac{1}{2}\sum_{m=0}^{n_b-1} 2^m - \frac{1}{2}\sum_{m=0}^{n_b-1} 2^m Z^{(m)}.
\end{equation}
It follows that
\begin{equation}
    \begin{aligned}
    D_{\mathrm{clock}}
    &= \exp\left(i\frac{\pi}{d}N\right) \\
    &= \exp\left(i\frac{\pi}{2d}\sum_{m=0}^{n_b-1} 2^m\right)
    \prod_{m=0}^{n_b-1} \exp\left(-i\frac{\pi}{2d} 2^m Z^{(m)}\right).
    \end{aligned}
\end{equation}
The prefactor
\begin{equation}
    \exp\left(i\frac{\pi}{2d}\sum_{m=0}^{n_b-1} 2^m\right)
\end{equation}
is a global phase and can be ignored. Hence $D_{\mathrm{clock}}$ decomposes as a product of $n_b$
single-qubit $Z$-rotations on the index register, one on each qubit $m$ with angle
$-\frac{\pi}{2d}2^m$. In particular, the cost of implementing $D_{\mathrm{clock}}$ is just $n_b$
single-qubit rotations.
\end{proof}

\begin{lemma}\label{prop:Dsign-implementation}
The diagonal operator
\begin{equation}
D_{\mathrm{sign}} = |0\rangle \langle 0| + \sum_{r=1}^{d-1} \mathrm{sgn}(c_r)\,|r\rangle\langle r|
\end{equation}
can be implemented using $4n_b$ T gates, together with $n_b$ scratch ancillas and a single flag ancilla.
\end{lemma}

\begin{proof}
From the explicit form of the Fourier coefficients of $\phi_x^2$, we have, for $r\neq 0$,
\begin{equation}
    \beta_r = c_r e^{i\pi r/d}, \qquad
    c_r = K\,\frac{\cos(\pi r/d)}{\sin^2(\pi r/d)}, \qquad K \in \mathbb{R},
\end{equation}
so $c_r\in\mathbb{R}$ for all $r$. For odd $d$, we have $0 < \frac{\pi r}{d} < \pi$ for $1 \le r \le
d-1$, so $\sin(\pi r/d) > 0$ and the sign of $c_r$ is determined by $\cos(\pi r/d)$ alone. Since
\begin{equation}
\begin{aligned}
\cos(\pi r/d) &> 0 \quad \text{for } 1 \le r \le \frac{d-1}{2}, \\
\cos(\pi r/d) &< 0 \quad \text{for } \frac{d+1}{2} \le r \le d-1,
\end{aligned}
\end{equation}
it follows that
\begin{equation}\label{eq:cr_crossover}
\begin{aligned}
c_r &> 0 \quad \text{for } 1 \le r \le \frac{d-1}{2}, \\
c_r &< 0 \quad \text{for } \frac{d+1}{2} \le r \le d-1,
\end{aligned}
\end{equation}
and the pattern is antisymmetric $c_{d-r} = -c_r$. In particular, $\mathrm{sgn}(c_r)$ is a
single-threshold function of the integer label $r$ with threshold $T = (d+1)/2$.

We implement $D_{\mathrm{sign}}$ using a standard reversible comparator acting on the $n_b$-qubit
index register and a single ancilla. Let $U_{\mathrm{cmp}}$ be a comparator unitary that compares
the value $r$ on the index register with the fixed classical threshold $T$, and writes the result
into an ancilla qubit $|a\rangle$ initialised in $|0\rangle$:
\begin{equation}
    U_{\mathrm{cmp}}: \quad |r\rangle|0\rangle \;\longmapsto\; |r\rangle|f(r)\rangle,
\end{equation}
where
\begin{equation}
    f(r) =
    \begin{cases}
        0, & r < T,\\
        1, & r \ge T.
    \end{cases}
\end{equation}
Starting from $|r\rangle|0\rangle$, we perform the following three-step sequence:
\begin{equation}
    \begin{aligned}
    |r\rangle|0\rangle
    &\xrightarrow{U_{\mathrm{cmp}}}
    |r\rangle|f(r)\rangle
    \xrightarrow{Z_a}
    (-1)^{f(r)} |r\rangle|f(r)\rangle \\
    &\xrightarrow{U_{\mathrm{cmp}}^\dagger}
    (-1)^{f(r)} |r\rangle|0\rangle.
    \end{aligned}
\end{equation}
After this sequence, the ancilla is returned to $|0\rangle$, and the net action on the index register is
\begin{equation}
    |r\rangle \;\longmapsto\; (-1)^{f(r)} |r\rangle =
    \begin{cases}
        +\,|r\rangle, & r < T,\\
        -\,|r\rangle, & r \ge T,
    \end{cases}
\end{equation}
which is precisely the desired diagonal
\begin{equation}
     \sum_{r=1}^{d-1} (-1)^{f(r)} |r\rangle\langle r|
    = \sum_{r=1}^{d-1} \mathrm{sgn}(c_r) |r\rangle\langle r|.
\end{equation}
In compact form we can write
\begin{equation}
    D_{\mathrm{sign}} = U_{\mathrm{cmp}}^\dagger \,(I \otimes Z_a)\, U_{\mathrm{cmp}}.
\end{equation}

The cost of this implementation is dominated by the comparator.
Ref.~\cite{hardy2024optimizedquantumsimulationalgorithms}
(Appendix~C.1(a), Eqs.~(C22)--(C23)) presents an $n_b$-bit comparator and shows that an optimized variant $\mathrm{CMP}'$
can be implemented using $n_b$ logical ANDs, hence a T-count of $4n_b$ (up to Cliffords), while
$\mathrm{CMP}'^\dagger$ is implemented using measurements and classically conditioned Cliffords.
Therefore the phase-kickback
construction $D_{\mathrm{sign}} = U_{\mathrm{cmp}}^\dagger (I\otimes Z_a) U_{\mathrm{cmp}}$ has T-count $4n_b$,
using $O(n_b)$ scratch qubits for carries and one flag ancilla.

\end{proof}

\section{Regime-2 Crossover Details}\label{app:reg2-crossover}
\begin{corollary}\label{prop:D-diagonal-cost-count}
Define $D$ as in Corollary~\ref{prop:D-diagonal-cost}.
Then there exist angles $\vartheta_0,\dots,\vartheta_{d-2}$ (unique modulo $4\pi$) and a global
phase $\gamma\in\mathbb{R}$ (unique modulo $2\pi$)
such that
\begin{equation}
D=e^{i\gamma}\prod_{k=0}^{d-2}R_Z^{(k,k+1)}(\vartheta_k).
\end{equation}
Consequently, $D$ can be implemented using at most $d-1$ adjacent embedded $SU(2)$ rotations.
\end{corollary}

\begin{proof}
Adjoin the trivial phase on $|0\rangle$ and write
\begin{equation}
D=\sum_{n=0}^{d-1}e^{i\theta_n}|n\rangle\!\langle n|,\qquad \theta_0=0,
\end{equation}
where for $r\ge 1$ the phases are those defining $D$ from Eq.~\ref{eq:beta_phases}.
With Eq.~\ref{eq:cr_crossover} determining $(\operatorname{sgn}(c_r))$, we have
\begin{equation}
\theta_r =
\begin{cases}
\displaystyle \frac{\pi r}{d}, & \operatorname{sgn}(c_r)=+1,\\[4pt]
\displaystyle \frac{\pi r}{d}+\pi, & \operatorname{sgn}(c_r)=-1,
\end{cases}
\end{equation}
    so that $\theta_r\in[0,2\pi)$ and $e^{i\theta_r}=\frac{\beta_r}{|\beta_r|}$.
Define the global phase
\begin{equation}
\gamma=\frac1d\sum_{n=0}^{d-1}\theta_n,\qquad \alpha_n=\theta_n-\gamma,
\end{equation}
so that $\sum_{n=0}^{d-1}\alpha_n=0$, and
\begin{equation}
D=e^{i\gamma}\sum_{n=0}^{d-1}e^{i\alpha_n}|n\rangle\!\langle n|.
\end{equation}
With this form, we can invoke Lemma~\ref{prop:su2_irred}: for any diagonal phases $(\alpha_n)_{n=0}^{d-1}$ with
$\sum_n\alpha_n=0$, there exists a unique choice of adjacent angles
\begin{equation}\label{eq:vartheta}
\vartheta_k=-2\sum_{n=0}^{k}\alpha_n,\qquad k=0,1,\dots,d-2,
\end{equation}
such that
\begin{equation}
\prod_{k=0}^{d-2}R_Z^{(k,k+1)}(\vartheta_k)=\sum_{n=0}^{d-1}e^{i\alpha_n}|n\rangle\!\langle n|.
\end{equation}
Therefore,
\begin{equation}
D=e^{i\gamma}\prod_{k=0}^{d-2}R_Z^{(k,k+1)}(\vartheta_k),
\end{equation}
proving existence; uniqueness ($\text{mod }4\pi$) follows from the uniqueness statement in
Lemma~\ref{prop:su2_irred}. In our case, by
Eq.~\ref{eq:beta_phases}, we have
\begin{equation}
e^{i\theta_r}=\frac{\beta_r}{|\beta_r|}
=\operatorname{sgn}(c_r)\,e^{i\pi r/d},
\qquad r=1,\dots,d-1,
\end{equation}
where $c_r$ is determined by Eq.~\ref{eq:cr_crossover}.
Substituting $\alpha_n=\theta_n-\gamma$ into Eq.~\ref{eq:vartheta} yields
\begin{equation}
\vartheta_k
=-2\Big(\sum_{n=0}^k\theta_n-(k+1)\gamma\Big),
\qquad k=0,1,\dots,d-2.
\end{equation}
Evaluating these sums gives the explicit closed form
\begin{equation}\label{eq:col14_cases}
\vartheta_k=
\begin{cases}
\displaystyle \frac{\pi}{d}(k+1)(4m-k),
& 0\le k\le m,\\[8pt]
\displaystyle \frac{\pi}{d}(k+1)(4m-k)\\
\displaystyle {}-2\pi(k-m),
& m<k\le d-2,
\end{cases}
\end{equation}
where $m =(d-1)/2$.
\end{proof}
Note that $R_Z^{(k,k+1)}(\vartheta_k)$ is trivial if and only if
$\vartheta_k\equiv 0\ (\mathrm{mod}\ 4\pi)$. In the numerical range considered
in Sec.~\ref{sec:regime2_nocs_crossover}, the exact value of $s(d)$ is
evaluated directly from the closed form in Eq.~\ref{eq:col14_cases}. For odd
$d\le 513$, one finds the decomposition uses
$d-1$, $d-2$, or $d-4$ adjacent embedded rotations. Therefore, Corollary~\ref{prop:D-diagonal-cost-count}
gives the uniform upper bound $d-1$.

\begin{lemma}\label{prop:prep-phi2-qudit-cost-nocs}
Define the state
\begin{equation}
|\psi_\beta\rangle
\;=\;
\sum_{r=1}^{d-1}\sqrt{\frac{|\beta_r|}{\Lambda}}\;|r\rangle,
\qquad
\Lambda=\sum_{r=1}^{d-1}|\beta_r|
\end{equation}
from Eq.~\ref{prep_def}. Then there exists a product of exactly $d-1$ embedded two-level $SU(2)$
rotations that prepare $|\psi_\beta\rangle$
from $|0\rangle$, i.e.,
\begin{equation}
\mathrm{PREP}\,|0\rangle = |\psi_\beta\rangle,
\end{equation}
with the explicit decomposition
\begin{equation}
\mathrm{PREP}
\;=\;
\prod_{r=d-1}^{1} R_Y^{(0,r)}(\theta_r).
\label{eq:prep-ry-product}
\end{equation}
\end{lemma}

\begin{proof}
Write
\begin{equation}
a_r = \sqrt{\frac{|\beta_r|}{\Lambda}},
\qquad r=1,\dots,d-1,
\end{equation}
so that $a_r\ge 0$ and $\sum_{r=1}^{d-1}a_r^2=1$.
For each $r\in\{1,\dots,d-1\}$, let
\begin{equation}
U_r = R_Y^{(0,r)}(\theta_r)
\end{equation}
denote the embedded two-level $Y$-rotation acting only on $\mathrm{span}\{|0\rangle,|r\rangle\}$.
In that two-dimensional subspace it acts as
\begin{equation}
\begin{aligned}
|0\rangle &\mapsto \cos\!\left(\frac{\theta_r}{2}\right)|0\rangle
+\sin\!\left(\frac{\theta_r}{2}\right)|r\rangle, \\
|r\rangle &\mapsto -\sin\!\left(\frac{\theta_r}{2}\right)|0\rangle
+\cos\!\left(\frac{\theta_r}{2}\right)|r\rangle,
\end{aligned}
\end{equation}
and it acts trivially on all basis states $|k\rangle$ with $k\neq 0,r$.
We apply these rotations sequentially to $|0\rangle$, i.e.
\begin{equation}
|\psi_r\rangle
= U_r U_{r-1}\cdots U_1\,|0\rangle.
\end{equation}
Since $U_r$ only mixes the amplitudes on $|0\rangle$ and $|r\rangle$, we can track the amplitudes explicitly.
After $r$ steps the state has the form
\begin{equation}
\begin{aligned}
|\psi_r\rangle
&=
\left(\prod_{k=1}^{r}\cos\!\left(\frac{\theta_k}{2}\right)\right)|0\rangle
\\
&\quad +
\sum_{j=1}^{r}
\left(
\sin\!\left(\frac{\theta_j}{2}\right)
\prod_{k=1}^{j-1}\cos\!\left(\frac{\theta_k}{2}\right)
\right) \\
&\qquad\times |j\rangle.
\end{aligned}
\label{eq:explicit-state}
\end{equation}
Indeed, at step $r$ the operator $U_r$ multiplies the current amplitude on $|0\rangle$ by
$\cos(\theta_r/2)$, and creates a new amplitude on $|r\rangle$ equal to
$\sin(\theta_r/2)$ times the \emph{current} amplitude on $|0\rangle$, while leaving all previously-prepared
levels $|1\rangle,\dots,|r-1\rangle$ unchanged.

We now choose the angles $\theta_r$ so that the amplitude on $|r\rangle$ equals the desired value $a_r$.
From \eqref{eq:explicit-state}, we require
\begin{equation}
a_r
=
\left(\prod_{k=1}^{r-1}\cos\!\left(\frac{\theta_k}{2}\right)\right)
\sin\!\left(\frac{\theta_r}{2}\right),
\end{equation}
or equivalently,
\begin{equation}
\sin\!\left(\frac{\theta_r}{2}\right)
=
\frac{a_r}{\displaystyle\prod_{k=1}^{r-1}\cos\!\left(\frac{\theta_k}{2}\right)}.
\label{eq:theta-recursion}
\end{equation}
This gives the explicit update rule for $\theta_r$.
This recursion is well-defined because the denominator is exactly the current amplitude on $|0\rangle$. Moreover, since
\begin{equation}
\left(\prod_{k=1}^{r-1}\cos^2\!\left(\frac{\theta_k}{2}\right)\right)
=
1-\sum_{j=1}^{r-1}a_j^2
\;\ge\; a_r^2,
\end{equation}
the ratio in \eqref{eq:theta-recursion} lies in $[0,1]$, so $\theta_r$ always exists.
With this choice, \eqref{eq:explicit-state} implies that after $d-1$ steps the amplitudes satisfy
\begin{equation}
|\psi_{d-1}\rangle
=
\left(\prod_{k=1}^{d-1}\cos\!\left(\frac{\theta_k}{2}\right)\right)|0\rangle
\;+\;
\sum_{r=1}^{d-1} a_r\,|r\rangle.
\end{equation}
Finally, since $\sum_{r=1}^{d-1}a_r^2=1$, the leftover amplitude on $|0\rangle$ must vanish, i.e.
\begin{equation}
\prod_{k=1}^{d-1}\cos\!\left(\frac{\theta_k}{2}\right)=0,
\end{equation}
and therefore
\begin{equation}
\left(\prod_{r=d-1}^{1} R_Y^{(0,r)}(\theta_r)\right)|0\rangle
=
\sum_{r=1}^{d-1} a_r\,|r\rangle
=
|\psi_\beta\rangle,
\end{equation}
which proves \eqref{eq:prep-ry-product}.
As $\beta_r\neq 0$ by Lemma~\ref{prop:qudit-pauli-expansion}, we have $a_r>0$ for all $r=1,\dots,d-1$,
so each level $|r\rangle$ must receive nonzero amplitude. Hence the construction requires (and uses)
exactly $d-1$ nontrivial embedded two-level rotations.
\end{proof}